\documentclass[journal,onecolumn,web]{ieeecolor}
\usepackage{generic}
\usepackage{comment}
\usepackage{derivative}
\usepackage{calligra}
\usepackage{graphicx} 
\usepackage{caption} 
\usepackage{subcaption}
\usepackage{amsmath,amssymb,bbm,color,psfrag,amsfonts,mathtools}
\usepackage{mathrsfs}
\usepackage{upgreek}
\usepackage{dsfont}
\usepackage{graphicx,cuted}                       % \includegraphics
\usepackage{wrapfig}
\usepackage{pdfpages}
\usepackage[ruled,lined]{algorithm2e} 
\usepackage{algpseudocode}% http://ctan.org/pkg/algorithmicx
\SetKwInOut{Input}{input}\SetKwInOut{Output}{output}
\usepackage{accents}
\usepackage{lipsum}
\usepackage[T1]{fontenc}
\usepackage{mathtools}
\usepackage{multicol}
\usepackage{paralist}   % \begin{inparaenum}[1)]
\usepackage{figlatex}
\let\labelindent\relax
\usepackage{enumitem}
	\usepackage{booktabs}  
	\usepackage{multirow}
	\usepackage[normalem]{ulem}
	
	\usepackage{caption}
	\usepackage{subcaption}
	
	\usepackage{cite}
	\usepackage[pdftex,pdfpagelabels,bookmarks,hyperindex,hyperfigures]{hyperref}
		\usepackage{cleveref}
\newtheorem{theorem}{Theorem}
\newtheorem{assumption}{Assumption}
\newtheorem{prop}{Proposition}

\newtheorem{lemma}{Lemma}
\newtheorem{remark}{Remark}
\newtheorem{example}{Example}

\newcommand{\ZZ}{\mathbb{Z}}
\newcommand{\RR}{\mathbb{R}}

\newcommand{\CC}{\mathbb{C}}
\newcommand{\D}{\mathcal{D}}
\newcommand{\C}{\mathcal{C}}
\newcommand{\state}{\phi}
\newcommand{\statehat}{\hat{\state}}
\newcommand{\im}{\mathbf{i}}
\newcommand*\diff{\mathop{}\!\mathrm{d}}
\newcommand{\abs}[1]{\left|#1\right|}
\newcommand{\wcomplex}{\upomega}
\newcommand{\kcomplex}{\upkappa}
\def\Xint#1{\mathchoice
	{\XXint\displaystyle\textstyle{#1}}%
	{\XXint\textstyle\scriptstyle{#1}}%
	{\XXint\scriptstyle\scriptscriptstyle{#1}}%
	{\XXint\scriptscriptstyle\scriptscriptstyle{#1}}%
	\!\int}
\def\XXint#1#2#3{{\setbox0=\hbox{$#1{#2#3}{\int}$}
		\vcenter{\hbox{$#2#3$}}\kern-.5\wd0}}

\def\dashint{\Xint-}
\begin{document}

\title{A complex spatial frequency approach to optimal control of finite-extent linear evolution systems}% with general boundary conditions}
\author{Zhexian Li, Athanassios S. Fokas, Ketan Savla \thanks{Z. Li is with the Sonny Astani Department of Civil and Environmental Engineering, University of Southern California, Los Angeles, USA. \texttt{zhexianl@usc.edu}. 
}
\thanks{
	A. S. Fokas is with the Department of Applied Mathematics and Theoretical Physics, University of Cambridge, Cambridge, UK, and Mathematics Research Center, Academy of Athens, Greece.
} 
\thanks{
	K. Savla is with the Sonny Astani Department of Civil and Environmental Engineering, Ming Hsieh Department of Electrical and Computer Engineering, and Daniel Epstein Department of Industrial and Systems Engineering, University of Southern California, Los Angeles, USA. K. Savla has a financial interest in Xtelligent, Inc.
}
}

\maketitle

\begin{abstract}
	We consider the linear quadratic regulator (LQR) for one-dimensional linear evolution partial differential equations (PDEs) on a finite interval in space. 
	The control is applied as an additive forcing term to PDEs.
	Existing methods for closed-form optimal control only apply to homogeneous (zero) boundary conditions, often resulting in series representations. 
	In this paper, we consider general smooth boundary conditions.
	We use the unified transform, namely the Fourier transform restricted to the bounded spatial domain, to decouple PDEs into a family of ordinary differential equations (ODEs) parameterized by \emph{complex} spatial frequency variables.
	Then, optimal control in the frequency domain is derived using LQR theory for ODEs.
	The inverse Fourier transform leads to non-causal terms in optimal control corresponding to integrals, over the real line, of future values of unspecified boundary conditions. 
	%in the original domain in terms of integrals on the real line depending on unknown boundary values in the future.
	To eliminate this non-causality, we deform the integrals to well-constructed contours in the complex plane along which the contribution of unknowns vanishes.
	For the reaction-diffusion equation, we show that the integral representation can be reformulated as a series representation, which leads to a state-feedback convolution form for optimal control, with the boundary conditions appearing as an additive term.
	%The dependence of optimal control on the given boundary conditions is separated from the state dependence as an additive term to the convolution.
	In numerical experiments,
	we illustrate the computational advantages of the integral representation in comparison to the series representation and structural properties of the convolution kernel.
	\end{abstract}
	
	
	\section{Introduction}
	 
	Evolution PDEs, where the time derivative is of first order, are widely used to describe dynamical systems in continuous space in control applications; this includes active flow control \cite{liu2024adjoint}, heat transfer \cite{lenhart1993optimal}, wind farms \cite{shapiro2017model}, and chemotaxis systems \cite{fister2003optimal}. 
	The boundary control problem where the boundary conditions are the control inputs has been studied extensively, mainly by leveraging the tool of backstepping; see \cite{smyshlyaev2010adaptive,ascencio2017backstepping}.
	In other types of control, the control action is applied internally over the entire spatial domain.
	In this paper, we will focus on the latter type of control problem.
	Optimal control in this case can be traced back to \cite{lions1971optimal}; following this important work, many numerical approaches have been developed for computing open-loop control, see for example \cite{hinze2008optimization}.
	These numerical approaches do not provide sufficient insight into the structure and design of feedback controllers.
	
	For finite-dimensional systems described by ODEs, it is possible to derive LQR control in closed form that leads to a state feedback controller and can be used for many control design tools \cite{anderson2007optimal}.
	Despite the fact that many control concepts can be naturally extended from finite-dimensional to infinite-dimensional systems \cite{curtain2012introduction}, there is no general recipe to derive closed-form optimal control for PDEs.
	For the special case of unbounded spatial domains, the authors of \cite{bamieh2002distributed} apply the Fourier transform in space to decouple a PDE into infinitely many ODEs parameterized by \emph{real} spatial frequency variables.
	Classical LQR theory for ODEs is used to derive closed-form optimal control in the real frequency domain.
	Then, the inverse Fourier transform gives rise to a convolution form of optimal control in the original spatial domain.
	The convolution form reveals fascinating distributed structures of optimal control. Since \cite{bamieh2002distributed}, there has been a great interest in deriving optimal control using frequency-domain methods, see for example \cite{d2003distributed,motee2008optimal}. 
	Unbounded spatial domains can be seen as approximations of large-scale systems, e.g., in \cite{jovanovic2005ill}. In this approximation, however, the effect of boundary conditions cannot be considered. 
	
	Among all frequency-domain methods for deriving closed-form optimal control of PDEs, only a few studies target finite-extent systems with bounded spatial domains; see, for example, \cite{langbort2005distributed,epperlein2016spatially}. 
	The former study assumes symmetric boundary conditions and applies spatial discretization to PDEs; the latter study considers continuous space, but only with homogeneous boundary conditions.
	Hence, none of them have solved the problem in continuous space with general boundary conditions.
	The main idea in these two studies is to embed finite-extent systems into equivalent spatially invariant systems to which classical Fourier methods apply. 
	The embedding technique is motivated by the method of images; this is a method for solving initial-boundary value problems of linear PDEs with symmetry properties. 
	As a result, the embedding technique in \cite{langbort2005distributed,epperlein2016spatially} is applied to systems that possess certain symmetry properties.
	
	Recently, a unified approach has been developed, also known as the Fokas method, which provides solutions to general linear and a class of nonlinear PDEs with general boundary conditions, see for example \cite{fokas1997unified,fokas2008unified,fokasbook,deconinck2014method}. This method has been applied to boundary control problems in \cite{kalimeris2023numerical}. In particular, the Fokas method provides a more general and simpler alternative to the method of images \cite[Section I.4.2]{fokas2008unified}. Furthermore, it has been integrated into the Mathematica symbolic solver \texttt{DSolveValue} \cite{Mathematica}. Traditionally, different types of linear evolution PDEs with different types of boundary conditions require the use of specialized methods. Taking the heat equation as an example, the sine transform and series are used for Dirichlet boundary value problems, whereas the cosine transform and series are used for Neumann boundary value problems. Unlike these specialized transforms or series, the Fokas method uses only the unified transform, namely the Fourier transform restricted to bounded spatial domains, for all different types of evolution PDEs and boundary conditions. Since the Fourier transform was initially used to derive optimal control for systems without boundaries in \cite{bamieh2002distributed}, it is natural to explore the control of finite-extent systems using the Fokas method. 
	
	The main difference between the classical Fourier transform and the unified transform is that the latter allows frequency variables to take \emph{complex} values.
	In the one-dimensional case, if the imaginary part of the spatial frequency is nonzero, the exponent in the spatial Fourier transform will go to infinity as the space variable goes to infinity.
	Restricting the spatial domain to a finite interval eliminates the possibility of the exponent going to infinity, and thus, the spatial frequency variable is well-defined in the entire complex plane. 
	This extension is particularly useful in deriving optimal control, as we will show that the closed-form expression involves integrals over spatial frequency along contours in the complex plane. 
	  
	The approach introduced in this paper consists of the following steps. 
	First, we transform a general linear evolution PDE into a system of decoupled ODEs parameterized by complex spatial frequency variables following the Fokas method.
	Each ODE is a scalar linear system with exogenous inputs involving both the given boundary conditions and unknown boundary values. 
	For example, when the Dirichlet boundary conditions are given, the Neumann boundary values are unknown, and vice versa.
	Then, we use the Hamilton-Jacobi-Bellman equation to derive closed-form optimal control for the ODEs with frequency variables on the real line.
	Applying the inverse Fourier transform to the optimal control of the ODEs, we obtain an integral representation on the real line for optimal control of the PDE; however, this depends on future values of unknown boundary conditions and thus is non-causal.
	Second, we present a general procedure to eliminate the dependence of unknown boundary values in the integral representation for optimal control.
	This is achieved by deforming the integral from the real line to well-constructed contours in the complex plane along which the contribution of unknown terms vanishes.
	We illustrate the procedure for the reaction-diffusion equation with Dirichlet boundary conditions.
	Third, we deform the complex contour back to the real line and obtain an equivalent series representation for optimal control of the reaction-diffusion equation.
	The series representation results in a state-feedback form of optimal control that is separated into a convolution kernel and an additive term; the convolution kernel specifies the dependence on the state at different locations, and the additive term involves the given boundary conditions.
	The kernel preserves a so-called Toeplitz plus Hankel structure, which was initially discovered in \cite{epperlein2016spatially} for the case of homogeneous boundary conditions.
	The structural properties of the kernel are demonstrated numerically for different PDE coefficients.
	We also conduct numerical experiments to compare the computational efficiency of the integral representation and the series representation.

	The paper is organized as follows: we formulate the LQR problem for linear evolution PDE in \Cref{sec:problem-formulation}.
	Then, we compare a classical approach and our approach to the heat equation with homogeneous boundary conditions as a motivating example in \Cref{sec:heat}. 
	\Cref{sec:integral-representation-real} derives an integral representation on the real line for optimal control. 
	\Cref{sec:integral-representation-complex} presents
	how to deform the integral representation from the real line to contours in the complex plane.
	For the particular case of the reaction-diffusion equation, \Cref{sec:feedback-form} rewrites the integral representation as a series representation, which leads to a state-feedback form of optimal control.
	\Cref{sec:numerical-experiment} reports numerical experiments to illustrate the computational advantages of the integral representation and structural properties of the state-feedback form. 
	\Cref{sec:conclusion} concludes our findings and presents future directions.
	
	We conclude this section by defining key notations to be used throughout the paper. 
	Let $\RR$ denote the real line. Let $\ZZ,\ZZ^+$ denote the set of integers and positive integers, respectively.
	Let $\CC,\CC^+,\CC^-$ denote the complex plane, the upper half-plane excluding $\RR$, and the lower half-plane excluding $\RR$, respectively.
	The imaginary unit is denoted by $\im := \sqrt{-1}$.
	Let $\mathcal{C}^\infty(\Omega,\CC)$ denote the space of infinitely differentiable complex functions on $\Omega$.
	Let $\mathcal{L}^2(\Omega,\RR)$ denote the space of square-integrable real functions in $\Omega$.
	Consider a complex-valued function $f(z)$, we use $\overline{f(z)}$ to denote the complex conjugate of $f(z)$, and $\text{Re}[f(z)]$ and $\text{Im}[f(z)]$ to denote the real part and the imaginary part of $f(z)$, respectively.
	To differentiate real and complex spatial frequency, we use both $\kcomplex$ and $k$ for frequency variables, where $\kcomplex\in\CC$ means that the frequency variable can take complex values, and $k\in\RR$ means that the frequency variable is restricted to real values.
	
	\section{Problem formulation}\label{sec:problem-formulation}
	Consider the linear evolution PDE in the domain $\Omega:=\{0<x<L,0<t<T\}$:
	\begin{equation}\label{eq:evolution-equation-unforced}
		\left(\partial_t + \sum_{j=0}^{n}\alpha_j(-\im \partial_x)^j\right)\state(x,t) = 0, 
	\end{equation}
	where the state $\state(x,t)$ is a complex scalar function, and $\alpha_j\in\CC, j=0,\ldots,n$ are complex coefficients with $n\in\ZZ^+,\alpha_n\neq0$. 
	For convenience, we write \eqref{eq:evolution-equation-unforced} in the form
	\begin{equation*}
		\state_t + w(-\im\partial_x)\state = 0,
	\end{equation*}
	where the differential operator $w(-\im\partial_x)$ is defined by the polynomial
	% admits the one-parameter family of solutions $\exp[\im \kcomplex x- w(\kcomplex)t]$ with parameter $\kcomplex\in\CC$.
	% Substituting this solution into \eqref{eq:evolution-equation-unforced} gives
		\begin{equation*}
			w(\kcomplex)=\alpha_n \kcomplex^n + \alpha_{n-1}\kcomplex^{n-1} + \ldots + \alpha_0,\quad\kcomplex\in\CC.
		\end{equation*}
	Let $\state(x,t)$ satisfy the initial condition $\state(x,0)=\state_0(x)$, where $\state_0(x)\in\mathcal{C}^{\infty}([0,L],\CC)$.
	Let $g_j(t),h_j(t)\in\mathcal{C}^\infty([0,T],\CC),j=0,1,\ldots,n-1,$ denote the boundary values at $x=0$ and $x=L$, respectively, i.e., $g_j(t)=\partial_x^j\state(0,t),h_j(t)=\partial_x^j\state(L,t)$.
	
	Following \cite{fokas2008unified}, we make the following assumption, so that an initial-boundary value problem of \eqref{eq:evolution-equation-unforced} is well-posed in $\mathcal{C}^\infty(\Omega,\CC)$.
	\begin{assumption}\label{assumption:well-posed}
		\begin{enumerate}
			\item The real part of $w(k)$ satisfies $\text{Re}[w(k)]\geq0$ for all $k\in\RR$. \label{assumption:well-posed-1}
		\item	\label{assumption:well-posed-2}
		$N$ boundary values are given at $x=0$ and $n-N$ boundary values are given at $x=L$, where
		\begin{equation*}
			N=\begin{cases}
				\frac{n}{2}, & \text{if } n \text{ is even},\\
				\frac{n+1}{2}, & \text{if } n \text{ is odd, }\text{Im}[\alpha_n] >0, \\
				\frac{n-1}{2}, & \text{if } n \text{ is odd, }\text{Im}[\alpha_n] <0.
			\end{cases}
		\end{equation*}
	\end{enumerate}
	\end{assumption}
	\begin{remark}
		Condition \ref{assumption:well-posed-1} in \Cref{assumption:well-posed} ensures that an initial value problem of \eqref{eq:evolution-equation-unforced} is well-posed for the case of unbounded spatial domain $-\infty<x<\infty$.
		The large $k$ limit of the condition $\text{Re}[w(k)]\geq0$ restricts the real part of the leading coefficient $\alpha_n$.
		If $n$ is even, then $k^n\geq0$ and condition \ref{assumption:well-posed-1} implies that $\text{Re}[\alpha_n]\geq0$.
		If $n$ is odd, then $\text{Re}[\alpha_n k^n]\to -\infty$ as $k\to \infty$ when $\text{Re}[\alpha_n]>0$ or as $k\to \infty$ when $\text{Re}[\alpha_n]<0$. 
		Therefore, condition~\ref{assumption:well-posed-1} requires $\text{Re}[\alpha_n]=0$ for odd $n$.
	
		Condition \ref{assumption:well-posed-2} specifies the number of boundary conditions required at the two endpoints in space.
		When $n=2$, this is consistent with the well-known Dirichlet boundary conditions where $g_0(t)=\state(0,t)$ and $h_0(t)=\state(L,t)$ are given, or Neumann boundary conditions where $g_1(t)=\partial_x\state(0,t)$ and $h_1(t)=\partial_x\state(L,t)$ are given.
	\end{remark}
	\begin{remark}
		In addition to Dirichlet and Neumann boundary conditions, linear combinations of $\partial_x^j\state(0,t)$ or $\partial_x^j\state(L,t)$ as boundary conditions, known as Robin boundary conditions, can also be used.
		Our method can be applied to all these types of boundary conditions.
	\end{remark}

	We consider control input to be a forcing term $u(x,t)$ added on the right-hand side of \eqref{eq:evolution-equation-unforced}, i.e.,
	\begin{equation}\label{eq:evolution-equation}
		\state_t + w(-\im\partial_x)\state = u(x,t).
	\end{equation}
	\Cref{assumption:well-posed} still implies the well-posedness of \eqref{eq:evolution-equation} provided that the forcing term $u(x,t)\in\mathcal{C}^{\infty}(\Omega,\CC)$ \cite{fokas2008unified}.
	\begin{example}
		\label{example:problem-formulation-even-order}
		Throughout the paper, 
		we will consider a class of even-order PDEs of the form,
		\begin{equation}\label{eq:even-order-pde}
			\state_t + (-\im \partial_x)^n\state = u(x,t),\quad w(\kcomplex) = \kcomplex^n,\quad n\text{ is even},
		\end{equation}
		with boundary conditions satisfying \Cref{assumption:well-posed}.
		\eqref{eq:even-order-pde} was analyzed in \cite{arbelaiz2024optimal} to illustrate their main results on optimal estimation.
	\end{example}
	\begin{example}\label{example:problem-formulation}
		Among physically meaningful equations, we will consider the reaction-diffusion equation,
			\begin{equation}\label{eq:heat}
				\state_t = \state_{xx} - c\state + u(x,t),\quad w(\kcomplex) = \kcomplex^2 + c, \quad c\geq0,
			\end{equation}
			with
			given Dirichlet boundary conditions $g_0(t)=\state(0,t),h_0(t)=\state(L,t).$
			Note that \eqref{eq:heat} does not fit into the class of even-order PDEs \eqref{eq:even-order-pde} when $c\neq0$.
			% The state, control, and boundary conditions are all real-valued since $\alpha_0=c,\alpha_1=0,\alpha_2=1$ for \eqref{eq:heat} are all real.
		% 	\item the simplified Korteweg-de Vries (KdV) equation:
		% 	\begin{equation}\label{eq:simplified-KdV}
		% 		\state_t + \state_{xxx} =u(x,t),\quad w(k) = - \im k^3
		% 	\end{equation}
		% given Dirichlet boundary conditions $\state(0,t)=g_0(t),\state(L,t)=h_0(t)$ and Neumann boundary conditions at one side $\state_x(L,t)=h_0(t)$.
	\end{example}
	% \begin{remark}
	% 		The original linearized KdV equation is given by $\state_t + \state_x +\state_{xxx}=u(x,t)$. Here we remove the first order spatial derivative $\state_x$ to reduce unnecessary complexity. The technical result for the simplified KdV equation can be extended to the original equation.
	% \end{remark}
	
	Our goal is to find optimal control $u^*(x,t)$ in the domain $\Omega$ that minimizes the objective functional 
	\begin{equation}\label{eq:obj-half}
		J = \int_0^T \int_{0}^{L}\left[\state^2(x,t)+u^2(x,t)\right]\diff x\ \diff t  + \int_0^L\state^2(x,T)\diff x,
	\end{equation}
	subject to dynamics \eqref{eq:evolution-equation} and corresponding initial and boundary conditions for which \eqref{eq:evolution-equation} is well posed. 
	In addition to finite $T$, we will also consider the infinite-time horizon case $T\to\infty$.
	We make the following assumption on boundary conditions so that the infinite-horizon control problem is well-posed.
	\begin{assumption}\label{assumption:boundary-conditions-infinite}
		There exists a finite time $\bar{t}$ such that $g_j(t),h_j(t)=0,j=0,\ldots,n-1,$ for all $t\geq\bar{t}$.
	\end{assumption}
	\begin{remark}
		Since $\bar{t}$ is arbitrary,
		\Cref{assumption:boundary-conditions-infinite} does not limit the practical choice of boundary conditions. We can augment the given boundary conditions defined in a finite time interval with a smooth function that vanishes after a finite time.
	\end{remark}
	
	A general principle for deriving optimal control is to first decouple the PDE into ODEs parameterized by frequency variables and then solve the resulting optimal control problem of the ODEs. 
	Parseval's identity is typically applied to ensure the equivalence between the objective in the original spatial domain and the objective in the frequency domain.
	For instance, the classical spatial Fourier transform is used in \cite{bamieh2002distributed} to decouple the PDE into ODEs for unbounded spatial domains.
	
	We will follow this principle throughout the paper and show that the unified transform introduced in \cite{fokas2008unified}, namely the spatial Fourier transform restricted to $0<x<L$, produces the appropriate analog to decouple the PDE \eqref{eq:evolution-equation} into ODEs.
	
	\section{Motivating example: the heat equation}
	Before considering the general PDE \eqref{eq:evolution-equation}, we first solve for optimal control of the heat equation, i.e., \eqref{eq:heat} with $c=0$, as a motivating example.
	Existing work focuses on the case of homogeneous boundary conditions \cite{epperlein2016spatially}. 
	We first discuss potential extensions from homogeneous to inhomogeneous boundary conditions using classical methods and the associated challenges.
	Then, we show that the unified transform overcomes these challenges.
	
	\label{sec:heat}
	\subsection{State transformation}
	\label{sec:state-transformation}
	A natural way is to use state transformation so that the transformed state satisfies homogeneous boundary conditions.
	Such a state transformation is a common tool used in boundary control problems; see \cite[Theorem 3.3.3]{curtain2012introduction},\cite{Ayamou2024finite}.
	For the heat equation with Dirichlet boundary conditions and $L=1$, a simple transformation is $\tilde{\state}(x,t)=\state(x,t) - (1-x)g_0(t) - xh_0(t)$, 
	under which it can be easily verified that $\tilde{\state}(0,t)=\tilde{\state}(1,t)=0$. 
	For $\tilde{\state}$, the heat equation is written as
	\begin{equation}\label{eq:heat-transformed}
		\tilde{\state}_t = \tilde{\state}_{xx} + u(x,t) -(1-x)\frac{\diff g_0(t)}{\diff t} - x\frac{\diff h_0(t)}{\diff t}.
	\end{equation}
	Note that $\tilde{\state}^2 = \state^2 - 2\state((1-x)g_0 + xh_0) + ((1-x)g_0 + xh_0)^2$ and the boundary conditions $g_0,h_0$ are not affected by control.
	Let $\Phi(x,t) = 2\tilde{\state}((1-x)g_0 + xh_0)$.
	Substituting $\tilde{\state}$ into the objective \eqref{eq:obj-half}, optimal control for \eqref{eq:obj-half} is equivalent to minimizing the objective 
	\begin{equation}\label{eq:objective-approximate}
		\begin{aligned}
		\tilde{J} = \int_0^T \int_{0}^{L}\left[\tilde{\state}^2(x,t)+u^2(x,t)\right]\diff x\, \diff t  + \int_0^L\tilde{\state}^2(x,T)\diff x
		+ \int_0^T\int_{0}^{L}\Phi(x,t)\diff x\,\diff t 
		+\int_0^L\Phi(x,T)\diff x,
	\end{aligned}
	\end{equation}
	subject to \eqref{eq:heat-transformed}. 
	Since $\Phi$ is a linear function of  $\tilde{\state}$, Parseval's identity does not apply.
	
	\subsection{Orthonormal basis}
	\label{sec:separation-variables}
	In order to ensure that Parseval's identity applies, another approach is to represent state and control using an orthonormal basis for the PDE under consideration.
	Consider the heat equation with Dirichlet boundary conditions and $L=\pi$, we use sine series representations $\state(x,t)=\sum_{m=1}^{\infty} \state_m(t)\sin(m x)$ and $u(x,t)=\sum_{m=1}^{\infty} u_m(t)\sin(m x)$.
	For inhomogeneous boundary conditions, the sine series representation is justified since sine functions form an orthonormal basis for $\mathcal{L}^2((0,L),\RR)$, i.e., defined on the \emph{open} interval $(0,L)$.
	Note that the series coefficients satisfy $\state_m(t)=\int_{0}^{L}\sin(m x)\state(x,t)\diff x$ and similarly for $u_m(t)$. Taking the time derivative of $\state_m(t)$ and substituting in the heat equation $\state_t = \state_{xx}+u$, from integration by parts we obtain
	\begin{equation}\label{eq:ode-separation-m}
		\frac{\diff\state_m}{\diff t} =- m^2\state_m + u_m + m\left[g_0(t) - (-1)^mh_0(t)\right],\ m\in\ZZ^+.
	\end{equation}
	This is a family of ODEs parameterized by $m$ with exogenous inputs involving boundary conditions $g_0(t),h_0(t)$. 
	Since $\sqrt{2/\pi}\sin(m x),m=1,2,\ldots$ is an orthonormal basis for $\mathcal{L}^2((0,\pi),\RR)$, Parseval's identity implies that minimizing \eqref{eq:obj-half} is equivalent to minimizing the infinite sum of the following objective from $m=1$ to $\infty$,
	\begin{equation}\label{eq:obj-separation}
		J_m = \int_0^T\left[\state_m^2(t)+u_m^2(t)\right]\diff t + \state_m^2(T),
	\end{equation}
	subject to \eqref{eq:ode-separation-m} that is decoupled in $m$.
	Optimal control $u_m^*(t)$ can be derived using the Hamilton-Jacobi-Bellman equation, see \Cref{sec:integral-representation-real}.
	 Then, optimal control for the heat equation is constructed as $u^*(x,t)=\sum_{m=1}^{\infty}u^*_m(t)\sin(m x)$, resulting in \emph{series representation}.
	% The key enabler for this approach is the Sturm-Liouville theorem, which only applies to second-order ODEs and homogeneous boundary conditions.
	
	Note that our ability to obtain the ODE in \eqref{eq:ode-separation-m} relies crucially on the analytical simplicity of the orthonormal basis for the PDE under consideration. While sinusoidal functions are well known basis for second-order PDEs, similar bases for higher-order PDEs are not known.  
	For instance, the analytical expression for the basis when $n=6$ is quite complex~\cite{papanicolaou2023orthonormal}, under which the form and analytical tractability of the counterpart of \eqref{eq:ode-separation-m} is unclear.
	
	% \begin{remark}
	% 	The method of separation of variables starts by finding an orthonormal basis of the space for a given PDE.
	% 	PDEs defined by self-adjoint spatial differential operators induce orthonormal bases of the Hilbert space where the PDE is defined \cite[Theorem A.4.20]{curtain2012introduction}.
	% 	However, it is hard to find such orthonormal bases for PDEs \emph{not} resulting in second-order ODEs in space with homogeneous boundary conditions, i.e., beyond the scope of the Sturm-Liouville theorem.
	% \end{remark}
	
	\subsection{The unified transform} 
	The approach described in \Cref{sec:separation-variables} requires finding the correct orthonormal basis for the PDE under consideration.
	The unified transform provides an alternative path without finding different orthonormal bases for different PDEs.
	To give a preview of our results, we revisit the heat equation with homogeneous boundary conditions and obtain closed-form optimal control in an integral representation using the unified transform. 
	
	The \textit{unified transform} in $\Omega$ is the Fourier transform restricted to the domain $0<x<L$ with frequency variable $\kcomplex$ allowed to take values in the complex plane $\CC$:
		 \begin{align}\label{eq:unified-transform}
				  \hat{f}(\kcomplex)&=\int_{0}^{L}e^{-\im \kcomplex x}f(x)\diff x, \quad \kcomplex\in\CC, \, \\ \label{eq:inverse-transform}
				  f(x)&=\int_{-\infty}^{\infty}\hat{f}(k)e^{\im kx}\frac{\diff k}{2\pi}, \quad 0< x<L.
		 \end{align}
	Note that the integral in the inverse transform \eqref{eq:inverse-transform} is along the real line and requires only the frequency variable $k\in\RR$.
	\begin{remark}
		The unified transform can be regarded as the classical direct and inverse Fourier transform applied to the following function,
		\begin{equation*}
			F(x)=\left\{
			\begin{aligned}
				&f(x), &0< x<L, \\
				&0, & \text{otherwise}.
			\end{aligned}
			\right.
		\end{equation*} 
		The frequency variable $k$ for the classical Fourier transform defined on the real line can only take real values, i.e., the imaginary part of $k$ is zero.
		The reason is the following.
		Consider a complex frequency $\kcomplex =\kcomplex_R + \im \kcomplex_I$, the exponential function $e^{-\im \kcomplex x}=e^{\kcomplex_I x}e^{-\im \kcomplex_R x}$ grows exponentially as $x\to\infty$ if $\kcomplex_I> 0$ or as $x\to-\infty$ if $\kcomplex_I<0$. 
	\end{remark}
	
	We will present two key steps in solving optimal control for the heat equation with homogeneous Dirichlet boundary conditions. 
	The following result for the heat equation is a special case of the results in \Cref{sec:integral-representation-real,sec:integral-representation-complex}.
	\subsubsection{Integral representation on the real line}
	Applying the unified transform to the heat equation and using integration by parts, the transformed state $\statehat(\kcomplex,t)$ can be shown to satisfy the ODE 
	\begin{equation}\label{eq:ode-heat-frequency}
			\hat{\state}_t(\kcomplex,t) = -\kcomplex^2\hat{\state}(\kcomplex,t) + \hat{u}(\kcomplex,t) + v(\kcomplex,t),\ \kcomplex \in \CC,
	\end{equation}
	with initial condition $\hat{\state}(\kcomplex,0)=\hat{\state}_0(\kcomplex)$, where $v(\kcomplex,t)=-g_1(t) - e^{-\im \kcomplex L}h_1(t)$ contains the unknown Neumann boundary values $g_1$ and $h_1$.
	
	Parseval's identity implies that minimizing \eqref{eq:obj-half} is equivalent to minimizing the integral of the following objective over $k$ from $-\infty$ to $\infty$:
	\begin{equation}\label{eq:obj-heat-frequency}
		J_k =\int_0^T\left[\hat{\state}^2(k,t)+\hat{u}^2(k,t)\right] \diff t\ + \hat{\state}^2(k,T).
	\end{equation}
	Since \eqref{eq:ode-heat-frequency} is decoupled in $\kcomplex$, we can solve for optimal control $\hat{u}^*(k,t)$ of \eqref{eq:ode-heat-frequency} with objective \eqref{eq:obj-heat-frequency} at each $k\in\RR$. Then,
	from $\hat{u}^*(k,t)$, we obtain optimal control $u^*(x,t)$ of the heat equation with objective \eqref{eq:obj-half} using the inverse transform \eqref{eq:inverse-transform}.
	In the infinite-horizon case $T\to\infty$, optimal control $\hat{u}^*(k,t)$ is given by
	\begin{equation*}
		\hat{u}^*(k,t) = -\hat{p}(k)\statehat(k,t) - \hat{p}(k)\int_{t}^{\infty} e^{\wcomplex(k)(t-\tau)}v(k,\tau)\diff \tau,
	\end{equation*}
	where $\wcomplex(k)=\sqrt{k^4+1}, \hat{p}(k)=-k^2 + \wcomplex(k)$.
	Substituting $\hat{u}^*(k,t)$ into \eqref{eq:ode-heat-frequency}, we obtain
	\begin{equation}\label{eq:state-heat-frequency}
			\statehat_t^* = -\wcomplex(k)\statehat^* + v(k,t) - \hat{p}(k)\int_{t}^{\infty} e^{\wcomplex(k)(t-\tau)}v(k,\tau)\diff \tau.
	\end{equation}
	Following \eqref{eq:ode-heat-frequency}, we have
	\begin{equation*}
		\hat{u}^*(k,t) = \statehat_t^* (k,t) +k^2\statehat^*(k,t) - v(k,t).
	\end{equation*}
	Applying the inverse transform \eqref{eq:inverse-transform} to the above equation, we obtain
	\begin{equation}\label{eq:optimal-control-heat}
		u^*(x,t) = \state_t^*(x,t) + \int_{-\infty}^{\infty} e^{\im k x} \left[k^2\statehat^*(k,t)-v(k,t)\right] \frac{\diff k}{2\pi},
	\end{equation}
	where both $\statehat(k,t)$ and $\state_t(x,t)$ can be evaluated by solving the ODE \eqref{eq:state-heat-frequency} and applying the inverse transform \eqref{eq:inverse-transform}.
	\subsubsection{Integral representation in the complex plane}
	The only remaining difficulty is that the function $v$ appearing in both \eqref{eq:state-heat-frequency} and \eqref{eq:optimal-control-heat} is unknown.
	To eliminate the unknown $v$, we employ tools from complex analysis to deform the integral in \eqref{eq:optimal-control-heat} from the real line to well-constructed contours in the complex plane. 
	Then, using the transformation $\kcomplex\to-\kcomplex$ in \eqref{eq:state-heat-frequency} and solving the resulting ODEs, we obtain two equations for two unknowns involving $g_1$ and $h_1$. Solving the two equations, these two unknowns are represented in terms of the given initial condition and a remaining unknown term involving $\statehat^*(\kcomplex,t)$ and $\statehat^*(-\kcomplex,t)$. The contribution of this remaining term vanishes along the constructed contours.
	Therefore, the resulting \emph{integral representation} of $u^*(x,t)$ only depends on the given initial condition and is given by
	\begin{equation}\label{eq:optimal-control-heat-integral}
		\begin{split}
			u^*(x,t) = -\int_{-\infty}^{\infty} e^{\im \kcomplex x-\wcomplex(k)t}\hat{p}(\kcomplex)\statehat_0(\kcomplex)\frac{\diff \kcomplex}{2\pi} + \int_{\partial\D^+} \frac{e^{-\wcomplex(\kcomplex)t}\hat{p}(\kcomplex)}{e^{\im \kcomplex L} - e^{- \im \kcomplex L}} \big[2\sin(\kcomplex x)\im e^{\im\kcomplex L}\statehat_0(\kcomplex) + 2\sin(\kcomplex (L-x))\im\statehat_0(-\kcomplex) \big] \frac{\diff\kcomplex}{2\pi},
		\end{split}
	\end{equation}
	where $\partial \D^+=\{\kcomplex\in\CC^+:\kcomplex=\abs{\kcomplex}e^{\im \theta}, \theta = \pi/8 \text{ or }7\pi/8\}$ serves as the desired contour.
	
	In \Cref{sec:feedback-form}, it will be shown that the integral representation for $u^*(x,t)$ in \eqref{eq:optimal-control-heat-integral} is equivalent to the series representation obtained from the approach discussed in \Cref{sec:separation-variables}.
	We return to the comparison between these two representations in \Cref{sec:feedback-form}.
	
	In \Cref{sec:integral-representation-real,sec:integral-representation-complex}, we generalize the two steps of deriving integral representations \eqref{eq:optimal-control-heat} and \eqref{eq:optimal-control-heat-integral}, i.e., on the real line and the complex plane, respectively, to PDE \eqref{eq:evolution-equation} with nonzero boundary conditions.
	
\section{Integral representation on the real line}\label{sec:integral-representation-real}

To simplify the notation, we omit the argument of functions when it is clear from the context.
We first consider the case when $T$ is finite and derive optimal control in the frequency domain.
The PDE \eqref{eq:evolution-equation} can be rewritten in the form of a family of PDEs parameterized by $\kcomplex\in\CC$:
\begin{equation}\label{eq:PDE-one-parameter}
	\begin{aligned}
		\left( e^{-\im \kcomplex x + w(\kcomplex)t} \state \right)_t - \left(e^{-\im \kcomplex x + w(\kcomplex)t} X \right)_x = e^{-\im \kcomplex x + w(\kcomplex)t} u, \quad \kcomplex\in\CC,
	\end{aligned}
\end{equation}
where 
\begin{equation}\label{eq:X-function}
	\begin{aligned}
		X(x,t,\kcomplex) =\im\left.\left(\dfrac{w(\kcomplex)-w(l)}{\kcomplex-l}\right)\right|_{l=-\im\partial_x}  \state(x,t).
	\end{aligned}
\end{equation}
To see the equivalence between \eqref{eq:PDE-one-parameter} and \eqref{eq:evolution-equation}, replacing $\state_t$ in \eqref{eq:PDE-one-parameter} with $-w(-\im\partial_x)\state + u$, we find 
$$
(w(\kcomplex) - w(-\im\partial_x))\state - (\partial_x - \im \kcomplex)X=0, 
$$
which yields \eqref{eq:X-function}.
Similarly, substituting \eqref{eq:X-function} into \eqref{eq:PDE-one-parameter} gives the original PDE \eqref{eq:evolution-equation}. 

Expanding \eqref{eq:X-function}, we can rewrite $X$ in the form
\begin{equation}\label{eq:X-function-expansion}
	X(x,t,\kcomplex) = \sum_{j=0}^{n-1}c_j(\kcomplex)\partial_x^j \state(x,t),
\end{equation}
where $\{c_j(\kcomplex)\}_{j=0}^{n-1}$ are known polynomials in $\kcomplex$ that can be explicitly computed, see for example the following.
\begin{example}\label{example:X-function-even-order}
	For the even-order PDE \eqref{eq:even-order-pde}, we have
	\begin{align*}
		X(x,t,\kcomplex) =\im\left.\dfrac{\kcomplex^n-l^n}{k-l}\right|_{l=-\im\partial_x}  \state(x,t)=\sum_{j=0}^{n-1}\im^{3j+1} \kcomplex^{n-1-j}\partial_x^j \state(x,t),
	\end{align*}
\end{example}
and $c_j(\kcomplex) = \im^{3j+1} \kcomplex^{n-1-j}$.
\begin{example}
\label{example-X-function}	
	For the reaction-diffusion equation \eqref{eq:heat}, we have
	\begin{align*}
		X(x,t,\kcomplex) =\im\left.\dfrac{\kcomplex^2-l^2}{k-l}\right|_{l=-\im\partial_x}  \state(x,t)=\state_x(x,t) + \im \kcomplex\state(x,t),
	\end{align*}
and $c_0(\kcomplex) = \im \kcomplex, c_1(\kcomplex) = 1$.
% 	For the simplified KdV equation \eqref{eq:simplified-KdV}, we have
% 		\begin{align*}
% 		X(x,t,k) &=\im(-\im)\left.\dfrac{k^3-l^3}{k-l}\right|_{l=-\im\partial_x}  \state(x,t)\\&=-\state_{xx}(x,t) - \im k\state(x,t) + k^2 \state(x,t),
% 	\end{align*}
% and $c_0(k) = k^2, c_1(k)=-\im k, c_2(k)=-1$.
\end{example}

Applying Green's theorem to the left-hand side of \eqref{eq:PDE-one-parameter} in the domain $[0,L]\times[0,t]$ yields
\begin{equation} \label{eq:green-theorem-applied}
	\begin{aligned}
		&\int_{0}^{L}\int_{0}^{t} \left( e^{-\im \kcomplex x + w(\kcomplex)\tau} \state \right)_\tau - \left(e^{-\im \kcomplex x + w(\kcomplex)\tau} X \right)_x \diff \tau\ \diff x\\
		=&\int_{0}^{L} e^{-\im \kcomplex x} \left[ e^{w(\kcomplex)t}\state(x,t)-\state_0(x)\right] \diff x  - \int_{0}^{t} e^{w(\kcomplex)\tau}\left[e^{-\im \kcomplex L}X(L,\tau,\kcomplex) - X(0,\tau,\kcomplex)\right]\diff \tau \\
		= &\int_{0}^{L}\int_{0}^{t} e^{-\im \kcomplex x + w(\kcomplex)\tau}u(x,\tau)\diff \tau\ \diff x.   
	\end{aligned}
\end{equation}
Let $\statehat, \statehat_0, \hat{u}$ denote the unified transform \eqref{eq:unified-transform} of $\state, \state_0, u$ respectively.
For all $j=0,\ldots,n-1$, let $\tilde{g}_j,\tilde{h}_j$ be the $t$-transform of $g_j,h_j$ respectively, defined as 
\begin{equation*}
	\tilde{g}_j(\kcomplex,t) := \int_{0}^{t}e^{\kcomplex\tau}g_j(\tau)\diff \tau,\quad \tilde{h}_j(\kcomplex,t) := \int_{0}^{t}e^{\kcomplex\tau}h_j(\tau)\diff \tau.
\end{equation*}  
Then, equation \eqref{eq:green-theorem-applied} can be rewritten as
\begin{equation}\label{eq:global-relation-u}
	\begin{aligned}
		e^{w(\kcomplex)t}\statehat(\kcomplex,t) &= \statehat_0(\kcomplex) + \int_{0}^{t}e^{w(\kcomplex)\tau}\hat{u}(\kcomplex,\tau)\diff \tau - \sum_{j=0}^{n-1}c_j(\kcomplex)\big[\tilde{g}_j(w(\kcomplex),t) - e^{-\im \kcomplex L}\tilde{h}_j(w(\kcomplex),t)\big].
	\end{aligned}
\end{equation}
We will call \eqref{eq:global-relation-u} the \emph{global relation} since it relates the transformed state $\statehat$ to initial and boundary conditions and control $\hat{u}$ in the frequency domain. It turns out that \eqref{eq:global-relation-u} is the solution to the following ODE,
\begin{equation}\label{eq:global-relation-u-ode}
	\begin{aligned}
		\statehat_t &= -w(\kcomplex)\statehat + \hat{u} + v(\kcomplex,t),
	\end{aligned}
\end{equation}
where 
\begin{equation}\label{eq:boundary-value-v}
	v(\kcomplex,t):=-\sum_{j=0}^{n-1}c_j(\kcomplex)\left[g_j(t) - e^{-\im \kcomplex L}h_j(t)\right].
\end{equation}

Parseval's identity implies that minimizing \eqref{eq:obj-half} is the same as finding $\hat{u}$ that minimizes
\begin{equation}\label{eq:obj-transform-half}
		\hat{J}=\int_{-\infty}^{\infty}\Big[\int_0^T\left( \statehat(k,t)\overline{\statehat(k,t)}+\hat{u}(k,t)\overline{\hat{u}(k,t)}\right)\diff t + \statehat(k,T)\overline{\statehat(k,T)}\Big]\diff k,
\end{equation}
subject to \eqref{eq:global-relation-u-ode} and $\statehat(k,0)=\statehat_0(k).$ Since \eqref{eq:global-relation-u-ode} is an ODE in the time variable decoupled in $k$, at a fixed $k\in\RR$ the optimal control problem is a finite-dimensional LQR problem in $\CC$ with exogenous inputs from boundary conditions.
Let $\boldsymbol{\statehat}:= [\statehat_{\text{Re}}\quad \statehat_{\text{Im}}]^\top,\boldsymbol{\hat{u}}:=[\hat{u}_{\text{Re}}\quad\hat{u}_{\text{Im}}]^\top,\boldsymbol{v}:=[v_{\text{Re}}\quad v_{\text{Im}} \quad 1]^\top$. Let $w(k)\equiv w_{\text{Re}}(k) +\im w_{\text{Im}}(k)$. The subscripts Re and Im represent the real and imaginary parts of the corresponding functions, respectively.
\begin{example}
	\label{ex:w-real-imaginary}
	For even-order PDE \eqref{eq:even-order-pde}, we have $w_{\text{Re}}(k) = k^n$ and $w_{\text{Im}}(k) = 0$.
	 For the reaction-diffusion equation \eqref{eq:heat}, we have $w_{\text{Re}}(k) = k^2+c$ and $w_{\text{Im}}(k) = 0$.
\end{example}

The following theorem provides the solution to the LQR problem.
\begin{theorem}\label{theo}
	The optimal control for \eqref{eq:obj-transform-half}  subject to \eqref{eq:global-relation-u-ode} is given by
	\begin{equation}\label{eq:u-opt-theo}
		\boldsymbol{\hat{u}}^*(k,t) = - \mathbf{P}(k,t)\boldsymbol{\statehat}(k,t) - \mathbf{R}(k,t)\boldsymbol{v}(k,t),
	\end{equation}
where for all $k\in\RR$, $\mathbf{P}(k,t)\in\RR^{2\times2}$ and $\mathbf{R}(k,t)\in\RR^{2\times3}$ are solutions to the following backward ODEs:
	\begin{align}
		-\mathbf{P}_t &= \mathbf{I - PBB^\top P + PA + A^\top P}, \quad\mathbf{P}(k,T) = \mathbf{I}, \label{eq:opt-law-p}\\
		% -\mathbf{Q}_t &=2\mathbf{Q}\,\mathbf{M}(k,t)+\mathbf{r}\,\mathbf{d}^\top - \dfrac{1}{4}\mathbf{r}\,\mathbf{r}^\top, \quad \mathbf{Q}(k,T) = \mathbf{0}_{2\times2} ,\\
		-\mathbf{R}_t&=\mathbf{PC + RD + A^\top R - P R}, \quad \mathbf{R}(k,T) = \mathbf{0}, \label{eq:opt-law-r}
	\end{align} 
	with
	\begin{align*}
		&\mathbf{A}(k) = \begin{bmatrix}
			-w_{\text{Re}}(k) & w_{\text{Im}}(k) \\ -w_{\text{Im}}(k) & -w_{\text{Re}}(k)
		\end{bmatrix},\quad \mathbf{B} = \begin{bmatrix}
			1 & 0 \\ 0 & 1
		\end{bmatrix},\ \mathbf{C} = \begin{bmatrix}
			1 & 0 & 0 \\ 0 & 1 &0 
		\end{bmatrix},\quad
		\mathbf{D}(k,t) = \begin{bmatrix}
			0 & 0 & \pdv{}{t}v_{\text{Re}}(k,t) \\ 0 & 0 & \pdv{}{t}v_{\text{Im}}(k,t) \\ 0 & 0 & 0
		\end{bmatrix}.
	\end{align*}
\end{theorem}

\begin{proof}
	The objective \eqref{eq:obj-transform-half} can be rewritten as
	\begin{align*}
		\hat{J}= & \int_{-\infty}^{\infty}\Big[\int_0^T\left( \statehat_{\text{Re}}^2(k,t) + \hat{u}_{\text{Re}}^2(k,t)\right)\diff t + \statehat_{\text{Re}}^2(k,T)\Big]\diff k +\int_{-\infty}^{\infty}\Big[\int_0^T\left( \statehat_{\text{Im}}^2(k,t) + \hat{u}_{\text{Im}}^2(k,t)\right)\diff t + \statehat_{\text{Im}}^2(k,T)\Big]\diff k.
	\end{align*}
	The dynamics \eqref{eq:global-relation-u-ode} can be rewritten as
	\begin{equation} \label{eq:global-relation-u-ode-real}
		\begin{aligned}
				\boldsymbol{\statehat}_t&=\mathbf{A}\boldsymbol{\statehat}+\mathbf{B}\boldsymbol{\hat{u}}+\mathbf{C}\boldsymbol{v},\\
				\boldsymbol{v}_t&=\mathbf{D}(t)\boldsymbol{v},
		\end{aligned}
	\end{equation}
	where the matrices are defined in \Cref{theo}.
	From now on, for brevity we omit the dependence on $k$ in this proof for brevity.
	The optimal control problem now is to find $\boldsymbol{\hat{u}}^*(t):= [\hat{u}_{\text{Re}}(t)\quad \hat{u}_{\text{Im}}(t)]^\top, t\in[0,T]$, minimizing 
	\begin{equation*}
		V(\boldsymbol{\statehat}(t_0),\boldsymbol{v}(t_0),\boldsymbol{\hat{u}}(\cdot), t_0)
		:=\int_{t_0}^T\big( \|\boldsymbol{\statehat}\|^2+\|\boldsymbol{\hat{u}}^2\|\big)\diff\tau + \|\boldsymbol{\statehat}(T)\|^2,
	\end{equation*}
	with $t_0=0$ and 
	subject to \eqref{eq:global-relation-u-ode-real}.
	
	Let $\hat{u}[a,b]$ denote a function $\hat{u}(\cdot)$ restricted to the interval $[a,b]$. Given $t\in[0,T]$, let us also make the definition
	\begin{equation}\label{eq:value-function-define}
		V^*(\boldsymbol{\statehat},\boldsymbol{v},t) := \min_{\boldsymbol{\hat{u}}[t,T]}\ V(\boldsymbol{\statehat},\boldsymbol{v},\boldsymbol{\hat{u}},t).
	\end{equation}
 Two properties of the optimal value function \eqref{eq:value-function-define} \cite[Eq.2.2-5, Eq.2.3-7]{anderson2007optimal} applied to dynamics \eqref{eq:global-relation-u-ode-real} are stated as follows. 
\begin{lemma}
	The optimal value function $V^*$ satisfies the Hamilton-Jacobi-Bellman equation:
	\begin{equation}\label{eq:hjb}
		\pdv{}{t}V^*(\boldsymbol{\statehat},\boldsymbol{v},t) 
		= -\min_{\boldsymbol{\hat{u}}[t,T]}\Big\{\|\boldsymbol{\statehat}\|^2+\|\boldsymbol{\hat{u}}^2\| 
		+ \Big[\pdv{}{\mathbf{\boldsymbol{\statehat}}}V^*(\boldsymbol{\statehat},\boldsymbol{v},t)\Big]^\top \boldsymbol{\statehat}_t 
		+\Big[\pdv{}{\mathbf{\boldsymbol{v}}}V^*(\boldsymbol{\statehat},\boldsymbol{v},t)\Big]^\top \boldsymbol{v}_t
		\Big\},
	\end{equation}
with the boundary condition $V^*(\boldsymbol{\statehat}(T),\boldsymbol{v}(T),T) = \|\boldsymbol{\statehat}(T)\|^2$, and the control given by the minimization in \eqref{eq:hjb} is the optimal control $\hat{u}^*(t)$ at time $t$.
\end{lemma}
\begin{lemma}\label{lemm:quadratic}
	The optimal value function $V^*(\boldsymbol{\statehat},\boldsymbol{v},t)$ has the quadratic form
	\begin{equation*}
		V^*(\boldsymbol{\statehat},\boldsymbol{v},t)=
		\begin{bmatrix}
			\boldsymbol{\statehat} & \boldsymbol{v}
		\end{bmatrix}^\top
		\begin{bmatrix}
			\mathbf{P}(t) & \mathbf{R}(t) \\ \mathbf{R}^\top(t) & \mathbf{Q}(t)
		\end{bmatrix}
		\begin{bmatrix}
			\boldsymbol{\statehat} \\ \boldsymbol{v}
		\end{bmatrix},
	\end{equation*}
	where $\mathbf{P}(t)\in\RR^{2\times2}$ and $\mathbf{Q}(t)\in\RR^{3\times3}$ are  symmetric.
\end{lemma}

From \Cref{lemm:quadratic}, the left-hand side in \eqref{eq:hjb} becomes
\begin{equation}\label{eq:hjb-lhs}
	\pdv{}{t}V^*(\boldsymbol{\statehat},\boldsymbol{v},t)=
	\begin{bmatrix}
		\boldsymbol{\statehat} & \boldsymbol{v}
	\end{bmatrix}^\top
	\begin{bmatrix}
		\mathbf{P}_t & \mathbf{R}_t \\ \mathbf{R}^\top_t & \mathbf{Q}_t
	\end{bmatrix}
	\begin{bmatrix}
		\boldsymbol{\statehat} \\ \boldsymbol{v}
	\end{bmatrix}.
\end{equation}
The partial derivatives on the right-hand side in \eqref{eq:hjb} can be expressed as
\begin{align*}
	&\left[\pdv{}{\boldsymbol{\statehat}}V^*(\boldsymbol{\statehat},\boldsymbol{v},t)\right]^\top \boldsymbol{\statehat}_t + \left[\pdv{}{\boldsymbol{v}}V^*(\boldsymbol{\statehat},\boldsymbol{v},t)\right]^\top \boldsymbol{v}_t 
	 =\ 
	 2\begin{bmatrix}
		\boldsymbol{\statehat} & \boldsymbol{v}
	\end{bmatrix}^\top
	\begin{bmatrix}
		\mathbf{P}(t) & \mathbf{R}(t) \\ \mathbf{R}^\top(t) & \mathbf{Q}(t)
	\end{bmatrix}
	\begin{bmatrix}
		\mathbf{A}\boldsymbol{\statehat}+\mathbf{B}\boldsymbol{\hat{u}}+\mathbf{C}\boldsymbol{v} \\ \mathbf{D}(t)\boldsymbol{v}
	\end{bmatrix}.
\end{align*}
% where 
% \begin{align*}
% &\tilde{\mathbf{P}}=\begin{bmatrix}
% 	p_{11} + p_{22} & 2p_{12}   \\ 2p_{21} &p_{11} + p_{22}  \end{bmatrix},\\
% 	&\tilde{\mathbf{R}}=\begin{bmatrix}
% 	 r_{11} + r_{22} &r_{12} + r_{21} \\  r_{12} + r_{21} &r_{11} + r_{22}  \end{bmatrix},\\
% 	&\tilde{\mathbf{Q}}=\begin{bmatrix}
% 	 Q_{11} + Q_{22} &2Q_{12} \\  2Q_{21} &Q_{11} + Q_{22}
% 	\end{bmatrix}, \\
% 	&\mathbf{A}(t) = \begin{bmatrix}
% 		-w & 0 & 1 & 0 & 0 &0 \\ 
% 		0 & -w^* & 0 & 0 & 1 &0\\ 
% 		0 & 0 & 0 & v_t& 0 & 0 \\ 
% 		0 & 0 & 0 & 0  & 0 & 0 \\ 
% 		0 & 0 & 0 & 0 &0 &v_t \\
% 		0 & 0 & 0 & 0 &0 &0
% 	\end{bmatrix}
% 	,\mathbf{B} = \begin{bmatrix}
% 		1 & 0 \\ 0 & 1 \\ 0 & 0 \\ 0 & 0 \\ 0 & 0 \\ 0 & 0
% 	\end{bmatrix},\\
% \end{align*}
Substituting the above terms into the right-hand side in \eqref{eq:hjb} and completing the square, we find
\begin{equation}\label{eq:hjb-rhs-derive}
	\begin{aligned}
		&\|\boldsymbol{\statehat}\|^2+\|\boldsymbol{\hat{u}}^2\| + \left[\pdv{}{\boldsymbol{\statehat}}V^*(\boldsymbol{\statehat},\boldsymbol{v},t)\right]^\top \boldsymbol{\statehat}_t + \left[\pdv{}{\boldsymbol{v}}V^*(\boldsymbol{\statehat},\boldsymbol{v},t)\right]^\top \boldsymbol{v}_t \\
		=\
		&(\boldsymbol{\hat{u}} + \mathbf{B}^\top \mathbf{P}\boldsymbol{\statehat} + \mathbf{B}^\top\mathbf{R}\boldsymbol{v})^\top(\boldsymbol{\hat{u}} + \mathbf{B}^\top \mathbf{P}\boldsymbol{\statehat} + \mathbf{B}^\top\mathbf{R}\boldsymbol{v}) \\ 
		&+ \boldsymbol{\statehat}^\top(\mathbf{I - PBB^\top P + PA + A^\top P})\boldsymbol{\statehat} \\ 
		&+2\boldsymbol{\statehat}^\top(\mathbf{PC + RD + A^\top R - PBB^\top R} )\boldsymbol{v}\\
		& + \boldsymbol{v}^\top(\mathbf{R^\top C + QD - RBB^\top R})\boldsymbol{v}.
	\end{aligned}
\end{equation}
Since the only dependence on $\boldsymbol{\hat{u}}$ in the right-hand side of \eqref{eq:hjb-rhs-derive} is a nonnegative quadratic function, the minimization in \eqref{eq:hjb} is achieved by 
\begin{equation}\label{eq:u-opt}
	\boldsymbol{\hat{u}}^*(t) = -\mathbf{B}^\top \mathbf{P}(t)\boldsymbol{\statehat}(t) - \mathbf{B}^\top\mathbf{R}(t)\boldsymbol{v}(t).
\end{equation}
Substituting \eqref{eq:u-opt} into \eqref{eq:hjb-rhs-derive} and equating \eqref{eq:hjb-lhs}, 
together with the boundary condition $V^*(\boldsymbol{\statehat}(T),\boldsymbol{v}(T),T) = \|\boldsymbol{\statehat}(T)\|^2$, 
we obtain \eqref{eq:opt-law-p}--\eqref{eq:opt-law-r}.
% \begin{align*}
% 	-\hat{p}_t &= 1 - 2k^2\hat{p} - \hat{p}^2,\quad \hat{p}(T) =  1\\
% 	-\mathbf{Q}_t &=2\mathbf{Q}\,\mathbf{M}(t)+\mathbf{r}\,\mathbf{d}^\top - \dfrac{1}{4}\mathbf{r}\,\mathbf{r}^\top, \mathbf{Q}(T) = \mathbf{0}_{2\times2} \\
% 	-\mathbf{r}_t&=-k^2\,\mathbf{r} + \mathbf{M}(t)^\top\mathbf{r}- \hat{p}\,\mathbf{r} + 2\hat{p}\,\mathbf{d}, \mathbf{r}(T) = \mathbf{0}_{2\times1}
% \end{align*} 
\end{proof}

	In the case of infinite-time control, i.e., $T\to\infty$, we disregard the terminal cost $\|\boldsymbol{\statehat}(k,T)\|^2$. 
	Then, the terminal condition for $\mathbf{P}$ becomes $\mathbf{P}(k,T)=\mathbf{0}$.
	Assumptions~\ref{assumption:well-posed}--\ref{assumption:boundary-conditions-infinite} ensure that objective \eqref{eq:obj-transform-half} is finite since \eqref{eq:global-relation-u-ode} is controllable and $v$ vanishes after finite time.
	% Note that the ODE \eqref{eq:opt-law-p} is also the differential Riccati equation for $\mathbf{\boldsymbol{\tilde{\state}_t}= A\boldsymbol{\tilde{\state}}+B\boldsymbol{\tilde{u}}}$.
	Following \cite[Section 6.2]{liberzon2011calculus}, the limit of the solution to \eqref{eq:opt-law-p} as $T\to\infty$ exists since the pair $(\mathbf{A},\mathbf{B})$ is controllable.
	 For every real $k$ and every finite $t$, the solution to backward ODE \eqref{eq:opt-law-p} starting from $\mathbf{P}(k,T)=\mathbf{0}$ with $T\to\infty$ converges to an equilibrium with nonnegative real entries, which is the positive root to the following algebraic Riccati equation,
	\begin{equation*}
		\mathbf{I - P^2 + PA + A^\top P}=0,
	\end{equation*}
	and the root is given by 
	\begin{equation*}
		\mathbf{P}(k) = \begin{bmatrix}
			\hat{p}(k)& 0 \\ 0 & \hat{p}(k)
		\end{bmatrix},
	\end{equation*}
	where $\hat{p}(k) = -w_{\text{Re}}(k) + \sqrt{w_{\text{Re}}^2(k) + 1}$.

	To solve $\mathbf{R}(k,t)$ in \eqref{eq:opt-law-r}, we first rewrite $\mathbf{R\equiv[\tilde{R}\quad r]}$ with $\mathbf{\tilde{R}}\in\RR^{2\times2},\mathbf{r}\in\RR^2$. Then, \eqref{eq:opt-law-r} can be rewritten as
	\begin{align}
		-\mathbf{\tilde{R}_t} &=\mathbf{ (A^\top-P) \tilde{R} + P},\quad \mathbf{\tilde{R}}(k,T)=0,\label{eq:ode-R-tilde}\\ 
		-\mathbf{r_t }& =\mathbf{ (A^\top-P) r + \tilde{R}}\begin{bmatrix}\pdv{}{t}v_{\text{Re}}(k,t)  \\ \pdv{}{t}v_{\text{Im}}(k,t)\end{bmatrix},\quad \mathbf{r}(k,T)=0  .\label{eq:ode-r}
	\end{align} 
	The solution to \eqref{eq:ode-R-tilde} is given by \footnote{$\mathbf{A^\top-P}$ is invertible since $\det(\mathbf{A^\top-P})=w_{\text{Re}}^2(k) + w_{\text{Im}}^2(k) + 1>0$.}
	\begin{equation}\label{eq:ode-R-tilde-solution}
		\begin{aligned}
			\mathbf{\tilde{R}}(k,t) &= \int_{t}^{T}e^{-(\mathbf{A^\top-P})(t-\tau)}\mathbf{P}\diff\tau\\
			&=(e^{-(\mathbf{A^\top - P})(t-T)} - \mathbf{I})(\mathbf{A^\top-P})^{-1}\mathbf{P}. \\
		\end{aligned}
	\end{equation}
	From \cite[Example 1]{bernstein1993some}, the matrix exponential above has the explicit form,
	\begin{equation*}
		\begin{aligned}
			e^{-(\mathbf{A^\top - P}) t} &= \begin{bmatrix}
				\cos(w_{\text{Im}}t) & \sin(w_{\text{Im}}t) \\ -\sin(w_{\text{Im}}t) & \cos(w_{\text{Im}}t)
			\end{bmatrix}e^{(\hat{p} + w_{\text{Re}})t}.
		\end{aligned}
	\end{equation*}
	Since $\hat{p}(k) + w_{\text{Re}}(k) = \sqrt{w_{\text{Re}}^2(k) + 1}>0$, \eqref{eq:ode-R-tilde-solution} converges to the equilibrium $\mathbf{\tilde{R}}^{\text{eq}}(k) = \mathbf{-(A^\top-P)}^{-1}\mathbf{P}$ as $T\to\infty$.
The solution to \eqref{eq:ode-r} is given by
\begin{equation*}
	\begin{aligned}
		\mathbf{r}(k,t) &= \int_{t}^{T}e^{-(\mathbf{A^\top-P})(t-\tau)}\mathbf{\tilde{R}}(k,t)\begin{bmatrix}\pdv{}{\tau}v_{\text{Re}}(k,\tau)  \\ \pdv{}{\tau}v_{\text{Im}}(k,\tau)\end{bmatrix}\diff\tau. \\
	\end{aligned}
\end{equation*}
Substituting \eqref{eq:ode-R-tilde-solution} into the above expression and taking the limit as $T\to\infty$, we find
\begin{equation*}
	\begin{aligned}
		\mathbf{r}(k,t) &= \int_{t}^{\infty}e^{-(\mathbf{A^\top-P})(t-\tau)}\mathbf{\tilde{R}}^{\text{eq}}(k)\begin{bmatrix}\pdv{}{\tau}v_{\text{Re}}(k,\tau)  \\ \pdv{}{\tau}v_{\text{Im}}(k,\tau)\end{bmatrix}\diff\tau. \\
	\end{aligned}
\end{equation*}

Let $\wcomplex(k):=w(k) + \hat{p}(k)=\sqrt{w_{\text{Re}}^2(k) + 1} + \im w_{\text{Im}}(k)$ and $\bar{\wcomplex}(k):=\sqrt{w_{\text{Re}}^2(k) + 1} - \im w_{\text{Im}}(k)$.
Using integration by parts and substituting the above expressions into $\mathbf{r}(k,t)$, we find
\begin{align*}
	\boldsymbol{c}^\top \mathbf{R}(k,t)\boldsymbol{v}(k,t)
	=\hat{p}(k)\int_{t}^{\infty}e^{\bar{\wcomplex}(k)(t-\tau)}v(k,\tau) \diff\tau,
	% &-\hat{p}(k)\int_{t}^{\infty}e^{(t-\tau)\sqrt{w_{\text{Re}}^2(k)+1}}2\im\sin(w_{\text{Im}}(k)(t-\tau))v(k,\tau) \diff\tau
\end{align*}
where $\boldsymbol{c}=[1 \quad \im]^\top$.
Using the identity $\hat{u}^*(k,t)=\boldsymbol{c}^\top\boldsymbol{\hat{u}}^*(k,t)$ and substituting $\mathbf{P}(k)$ and the above expression into \eqref{eq:u-opt-theo}, the infinite-time optimal control law in the frequency domain can be expressed as
\begin{equation}\label{eq:u-opt-transform}
	\hat{u}^*(k,t) = -\hat{p}(k)\statehat(k,t) - \hat{p}(k)\int_{t}^{\infty}e^{\bar{\wcomplex}(k)(t-\tau)}v(k,\tau) \diff\tau.
\end{equation}
% Applying the inverse Fourier transform to \eqref{eq:u-opt-transform}, we find the optimal control in the original space domain as follows:
% \begin{multline}[]
% 		u^*(x,t) = -\int_{-\infty}^{\infty}e^{\im kx}\hat{p}(k) \Big[\statehat(k,t)  \\
% 		-\int_{t}^{\infty}e^{\bar{\wcomplex}(k)(t-\tau)}v(k,\tau)\diff\tau\Big] \frac{\diff k}{2\pi}.
% \end{multline}
Substituting \eqref{eq:u-opt-transform} into \eqref{eq:global-relation-u-ode}, the closed-loop system in the frequency domain is given by 
\begin{equation}\label{eq:global-relation-u-ode-control}
	\statehat_t^* = -\wcomplex(k)\statehat^* + v(k,t) -
	\hat{p}(k)\int_{t}^{\infty}e^{\bar{\wcomplex}(k)(t-\tau)}v(k,\tau) \diff\tau.
\end{equation}
% Now, we analyze the well-posedness and stability of the heat equation \eqref{eq:evolution-equation} after substituting optimal control \eqref{eq:u-opt-transform}.
% Unlike traditional methods that represent PDEs in abstract differential equations, e.g., see \cite{wang2014boundary}, the Fourier transform decomposes the heat equation into an ODE \eqref{eq:global-relation-u-ode-control} for each $k$.
% Therefore, it suffices to analyze the well-posedness and stability of the ODE for each $k$.
% Let $b(k,t):=v(k,t) - \hat{p}(k)\int_t^{\infty}e^{\wcomplex(k)(t-\tau)}v(k,\tau)\diff\tau$. 
% The solution to \eqref{eq:global-relation-u-ode-control} with the transformed initial condition $\statehat(k,0)=\statehat_0(k)$ is given by
% \begin{equation*}
% 	\statehat(k,t) = e^{-\wcomplex(k)t}\statehat_0(k) + \int_{0}^{t}e^{-\wcomplex(k)(t-\tau)}b(k,\tau)\diff\tau.
% \end{equation*}
% For simplicity, we assume that $v(\cdot,\cdot)$ is bounded, and there exists a finite time $K$ such that $v(k,t) = 0$ for all $t>K$. Then, $b(k,t)=0$ for all $t>K$.
% Following \Cref{re:real-k}, for every real $k$, the solution $\statehat(k,t)$ is unique, and $\statehat(k,t)\to0$ as $t\to\infty$ since $\wcomplex(k)>0$.
% Tighter conditions for well-posedness and stability can be obtained through more complicated analysis and will be pursued in the future.
%\begin{remark}
%	At every time $t$, the optimal control does not depend on the boundary conditions $g_0(\tau), g_1(\tau)$ and state $\statehat(k,\tau)$ in the future when $\tau>t.$ 
%\end{remark}
Combining \eqref{eq:u-opt-transform} and \eqref{eq:global-relation-u-ode-control}, we obtain
\begin{equation*}
	\hat{u}^*(k,t) = \statehat_t^* + w(k) \statehat^* - v(k,t).
\end{equation*}
Applying the inverse Fourier transform, we find the following expression for $u^*(x,t)$,
\begin{equation}\label{eq:u-opt-original} 
		u^*(x,t) = u^*_1(x,t) + u^*_2(x,t) + u^*_3(x,t),
\end{equation}
where 
\begin{equation*}
	\begin{aligned}
		 u^*_1(x,t):= \state_t^*(x,t),\quad
		 u^*_2(x,t) := \int_{-\infty}^{\infty}e^{\im kx}w(k)\statehat^*(k,t)\frac{\diff k}{2\pi}, \quad 
		 u^*_3(x,t) := -\int_{-\infty}^{\infty}e^{\im kx}v(k,t)\frac{\diff k}{2\pi}.
	\end{aligned}
\end{equation*}
Note that $u^*_1(x,t)$ can be evaluated by taking the time derivative of $\int_{-\infty}^{\infty}e^{\im kx}\statehat^*(k,t)\frac{\diff k}{2\pi}$,
where the integral can be seen as a special case of $u^*_2(x,t)$ with $w(k)=1$.
Therefore, it suffices to focus on $u^*_2(x,t)$ and $u^*_3(x,t)$. Also, 
\eqref{eq:optimal-control-heat} for the heat equation can be obtained from \eqref{eq:u-opt-original} by substituting $w(k)=k^2$.

The transformed state $\statehat^*(k,t)$ can be represented in terms of the initial and boundary values using \eqref{eq:global-relation-u-ode-control}.
From \Cref{assumption:well-posed}, only $n$ boundary values are given, whereas $2n$ boundary values are needed for $v$ in \eqref{eq:boundary-value-v}. Thus, \eqref{eq:u-opt-original} cannot be computed yet with $n$ unknown boundary values.
In the following section, we will derive an expression for $u^*_2(x,t)$ that depends only on the given initial and boundary conditions.
This is achieved by extending the frequency domain in \eqref{eq:u-opt-original} from $k\in\RR$ to $\kcomplex\in\CC$.
Moreover, the third integral $u^*_3(x,t)$ that involves $2n$ boundary values in $v$ will be shown to be zero. 

\section{Integral representation in the complex plane}\label{sec:integral-representation-complex}
In this section, we present a general procedure for removing the dependence of optimal control \eqref{eq:u-opt-original} on unknown boundary values.
We first introduce three steps to derive an expression for $u^*_2(x,t)$ that depends only on the given initial and boundary values.
The three steps are applied to the class of even-order PDEs \eqref{eq:even-order-pde} for illustration.
Then, we show that $u^*_3(x,t)$ is equal to zero. 
% provided that for all $j=0,\ldots,n-1, c_j(\kcomplex)$ in \eqref{eq:X-function-expansion} is a polynomial of the variable $\im k$, i.e., $c_j(\kcomplex)=\sum_{l=0}^{n_j} \tilde{c}_{jl}(\im \kcomplex)^l$ with coefficients $\tilde{c}_{jl}\in\CC,l=0,\ldots,n_j$.
% This is true for \eqref{eq:even-order-pde}, see \Cref{example:X-function-even-order}, and many physically meaningful PDEs, including the reaction diffusion equation in \eqref{eq:heat}, see \Cref{example-X-function}. 
Finally, we apply the procedure to the reaction-diffusion equation that does not fit into \eqref{eq:even-order-pde}.
We write $\statehat=\statehat^*$ for brevity since we will only consider the closed-loop system under optimal control \eqref{eq:u-opt-transform}.

\begin{figure}[tb]
	\centering
	\includegraphics[width=0.3\linewidth]{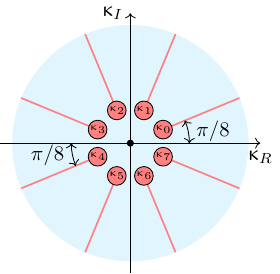}
	\caption{\sf Branch cuts for $\sqrt{\kcomplex^{2n} + 1}$ with $n=4$}
	\label{fig:branch-cut-heat}
\end{figure}

\subsection{General procedure}
\label{sec:general-procedures}
\subsubsection{The integral representation}
\label{sec:integral-representation}
We first obtain an integral representation for $u^*_2(x,t)$ in terms of the initial and boundary values.  
Solving \eqref{eq:global-relation-u-ode-control}, we find the following global relation under optimal control,
\begin{equation}\label{eq:global-relation-subsitute-u}
	\begin{aligned}
		\statehat(k,t) = e^{-\wcomplex(k)t} \statehat_0(k)- e^{-\wcomplex(k)t}\sum_{j=0}^{n-1}c_j(k)\Big[&\tilde{g}_j(\wcomplex(k),t) - e^{-\im k L}\tilde{h}_j(\wcomplex(k),t)\\
		&\hspace{0in}\!-\! \hat{p}(k)(\underline{g}_j(\wcomplex(k),\bar{\wcomplex}(k),t)\! -\! e^{-\im k L}\underline{h}_j(\wcomplex(k),\bar{\wcomplex}(k),t))\Big],
	\end{aligned}
\end{equation}
where for all $j=0,1,\ldots,n-1,$
\begin{equation*}
	\begin{aligned}
		\underline{g}_j(k,\tilde{k},t) := \int_{0}^{t}e^{k\tau}\left[\int_\tau^{\infty} e^{\tilde{k}(\tau-s)}g_j(s)\diff s\right] \diff\tau,\quad 
		\underline{h}_j(k,\tilde{k},t) := \int_{0}^{t}e^{k\tau}\left[\int_\tau^{\infty} e^{\tilde{k}(\tau-s)}h_j(s)\diff s\right] \diff\tau.
	\end{aligned}
\end{equation*}
Substituting \eqref{eq:global-relation-subsitute-u} into \eqref{eq:u-opt-original}, we have
\begin{equation}\label{eq:u-2-integral-real}
	\begin{aligned}
		u^*_2(x,t) = \int_{-\infty}^{\infty}e^{\im k x-\wcomplex(k)t}w(k)\statehat_0(k)\frac{\diff k}{2\pi}  - \int_{-\infty}^{\infty}e^{\im k x-\wcomplex(k)t}w(k)\tilde{g}(k,t)\frac{\diff k}{2\pi} + \int_{-\infty}^{\infty}e^{\im k(x-L)-\wcomplex(k)t}w(k)\tilde{h}(k,t) \frac{\diff k}{2\pi},
	\end{aligned}
\end{equation}
where 
$$\tilde{g}(k,t)=\sum_{j=0}^{n-1}c_j(k)\check{g}_j(k,t),\quad
\tilde{h}(k,t)=\sum_{j=0}^{n-1}c_j(k)\check{h}_j(k,t),$$
with
\begin{equation}\label{eq:check-g-h}
	\begin{aligned}
		\check{g}_j(k,t) = \tilde{g}_j(\wcomplex(k),t) -\hat{p}(k)\underline{g}_j(\wcomplex(k),\bar{\wcomplex}(k),t),\quad
		\check{h}_j(k,t) = \tilde{h}_j(\wcomplex(k),t) -\hat{p}(k)\underline{h}_j(\wcomplex(k),\bar{\wcomplex}(k),t).
	\end{aligned}
\end{equation}

\subsubsection{Contour deformation}
\label{sec:contour-deformation}

Next, we deform the integrals involving $\tilde{g}$ and $\tilde{h}$ in \eqref{eq:u-2-integral-real} from the real line to contours in $\CC$.
Note that $\wcomplex(k),\bar{\wcomplex}(k),$ and $\hat{p}(k)$ in \eqref{eq:global-relation-subsitute-u} are only defined when $k\in\RR$.
Recall that $\wcomplex(k)= \sqrt{w_{\text{Re}}^2(k) + 1} + \im w_{\text{Im}}(k)$ where functions $w_{\text{Re}}$ and $w_{\text{Im}}$ are real and imaginary parts of $w(k)$ when $k\in\RR$. 
Therefore, we first need to find the \emph{analytic continuation}\footnote{A function $g(z)$ defined on $A$ is the analytic continuation of $f(z)$ defined on $B$ if $g(z)$ and $f(z)$ are analytic and $f(z)=g(z)$ on common domain $A\cap B$.} in $\CC$ for $\sqrt{w_{\text{Re}}^2(k) + 1}, w_{\text{Re}}(k)$ and $w_{\text{Im}}(k)$ that define $\wcomplex(k),\bar{\wcomplex}(k),$ and $\hat{p}(k)$.
Since $w_{\text{Re}}(k)$ and $w_{\text{Im}}(k)$ are polynomials, $w_{\text{Re}}(\kcomplex)$ and $w_{\text{Im}}(\kcomplex)$ with $\kcomplex\in\CC$ are entire functions, i.e., analytic everywhere in $\CC$.
The square root $\sqrt{w_{\text{Re}}^2(\kcomplex) + 1}$ is analytic in $\CC$ except at branch cuts, e.g., the lines that join two branch points\footnote{Consider a square root function $\sqrt{f(z)},z\in\CC$. 
A point is a branch point if the multivalued function $\sqrt{f(z)}$ is discontinuous upon traversing a small circuit around this point.
Therefore, every point such that $f(z)=0$ is a branch point of $\sqrt{f(z)}$. Note that the point $z=\infty$ in $\CC$ is also a branch point \cite{ablowitz2003complex}.}.
Removing branch cuts from $\CC$ allows $\sqrt{w_{\text{Re}}^2(\kcomplex) + 1}$ to be single-valued and thus analytic. 

\begin{remark}
	Following \Cref{ex:w-real-imaginary}, for the even-order PDE \eqref{eq:even-order-pde}, we have $w_{\text{Re}}(\kcomplex)=\kcomplex^n$ and $w_{\text{Im}}(\kcomplex)=0$.
	Note that $ w_{\text{Re}}(\kcomplex)\neq\text{Re}[w(\kcomplex)]$ and $ w_{\text{Im}}(\kcomplex)\neq\text{Re}[w(\kcomplex)]$.
\end{remark}

\begin{example}
	\label{example:branch-cut-even-order}
	For the even-order PDE \eqref{eq:even-order-pde}, $\sqrt{w_{\text{Re}}^2(\kcomplex) + 1} = \sqrt{\kcomplex^{2n} + 1}$ has branch points at $\kcomplex_m=\exp[\im\frac{\pi+2m\pi}{2n}],m=0,1,\ldots,2n-1$. 
	An example of branch cuts is to connect every $\kcomplex_m$ with $\infty$ by the ray with the same argument as $\kcomplex_m$, see the red lines in \Cref{fig:branch-cut-heat}.
\end{example}
\begin{figure}[tb]
	\begin{subfigure}{0.49\textwidth}
		\centering
		\includegraphics[width=0.6\textwidth]{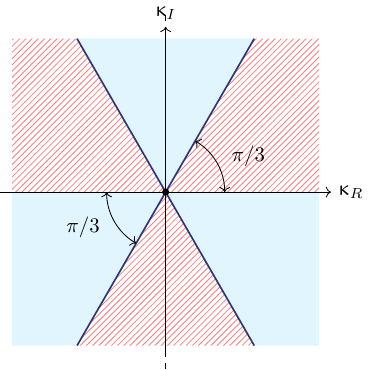} 
		\caption{$\text{Re}[\im \kcomplex^3]$}
	\end{subfigure}\hfill
	\begin{subfigure}{0.49\textwidth}
		\centering
		\includegraphics[width=0.6\textwidth]{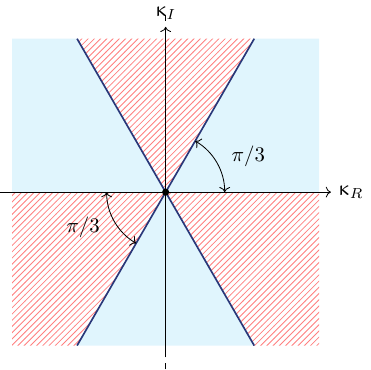}
		\caption{$\text{Re}[-\im \kcomplex^3]$}
	\end{subfigure}
	\caption{\sf
		Consider the PDE $\state_t=\state_{xxx}$ with $w_{\text{Re}}(\kcomplex)=0, w_{\text{Im}}(\kcomplex)=\im \kcomplex^3$.
		$\wcomplex(\kcomplex)\sim \im \kcomplex^3$ and $\bar{\wcomplex}(\kcomplex)\sim -\im \kcomplex^3$ for large $\kcomplex$. The blue regions represent positive real parts.
	}
	\label{fig:real-positive-region}
\end{figure}

In the infinite-time case, we also need to ensure that the integrands in \eqref{eq:u-2-integral-real} are bounded when extending from $k\in\RR$ to $\kcomplex\in\CC$ as $t\to\infty$. This requires that $\text{Re}[\wcomplex(\kcomplex)]\geq0$ under which $e^{-\wcomplex(\kcomplex)t}=e^{-\text{Re}[\wcomplex(\kcomplex)]t}e^{-\im\text{Im}[\wcomplex(\kcomplex)]t}$ is bounded. 
Moreover, we will also need to consider the case $\kcomplex\to\infty$ in contour deformation.
This additionally requires that $\text{Re}[\bar{\wcomplex}(\kcomplex)]\geq0$ so that $e^{\bar{\wcomplex}(\kcomplex)(\tau-s)}$ where $s>\tau$ is bounded.
Therefore, the integrals can only be deformed to contours in regions where $\text{Re}[\wcomplex(\kcomplex)]\geq0$ and $\text{Re}[\bar{\wcomplex}(\kcomplex)]\geq0$.
Furthermore, we require $\text{Re}[\wcomplex(\kcomplex)]>0$ to ensure that the integrands decay exponentially as $\kcomplex\to\infty$ to apply Jordan's Lemma in \Cref{lemm:jordan} later in contour deformation.

\begin{remark}\label{re:contour-deformation}
	If $w_{\text{Im}}(\kcomplex)=0$, then $\wcomplex(\kcomplex)=\bar{\wcomplex}(\kcomplex)$.
	If $w_{\text{Im}}(\kcomplex)\neq0$, then the regions where $\text{Re}[\wcomplex(\kcomplex)]\geq0$ and $\text{Re}[\bar{\wcomplex}(\kcomplex)]\geq0$ may not overlap, see \Cref{fig:real-positive-region} for illustration. 
	Optimal control when $w_{\text{Im}}(\kcomplex)\neq0$ will be pursued in the future.
\end{remark}

We consider the even-order PDE \eqref{eq:even-order-pde} with $\wcomplex(\kcomplex)=\bar{\wcomplex}(\kcomplex)=\sqrt{\kcomplex^{2n}+1}$ to illustrate contour deformation.
Using Cauchy's Theorem in \Cref{lemm:cauchy} and Jordan's Lemma in \Cref{lemm:jordan}, we will show that \eqref{eq:u-2-integral-real} is equal to 
\begin{equation}\label{eq:u-2-integral-complex}
	\begin{aligned}
		u^*_2(x,t) = \int_{-\infty}^{\infty}e^{\im \kcomplex x-\wcomplex(\kcomplex)t}w(\kcomplex)\statehat_0(\kcomplex)\frac{\diff \kcomplex}{2\pi} 
		  - \int_{\partial\D^+}e^{\im \kcomplex x-\wcomplex(\kcomplex)t}w(\kcomplex)\tilde{g}(\kcomplex,t)\frac{\diff \kcomplex}{2\pi} 
		 +\int_{\partial\D^-}e^{\im \kcomplex(x-L)-\wcomplex(\kcomplex)t}w(\kcomplex)\tilde{h}(\kcomplex,t) \frac{\diff \kcomplex}{2\pi},
	\end{aligned}
\end{equation}
where $\partial\D^+$ and $\partial\D^-$ are boundaries of $\D^+$ in $\CC^+$ and $\D^-$ in $\CC^-$, respectively, and $\D^+$ and $\D^-$ are regions such that $\wcomplex(\kcomplex)$ is analytic and $\text{Re}[\wcomplex(\kcomplex)]>0,\text{Re}[\bar{\wcomplex}(\kcomplex)]\geq0$ in both $\CC^+\setminus\D^+$ and $\CC^-\setminus\D^-$. 

Consider the contour $\partial\D^+=\{\kcomplex\in\CC^+:\kcomplex=|\kcomplex|e^{\im\theta},\theta=\beta^+ \text{ or }\pi-\beta^+\}$ and $\partial\D^-=\{\kcomplex\in\CC^-:\kcomplex=|\kcomplex|e^{\im\theta},\theta=\beta^- \text{ or }-\beta^- - \pi\}$ as shown in \Cref{fig:contour-deformation}. 
Following the branch cuts constructed in \Cref{example:branch-cut-even-order}, we choose $\beta^+\in(0,\frac{\pi}{2n}),\beta^-\in(-\frac{\pi}{2n},0)$.
Then, $\wcomplex(\kcomplex)=\bar{\wcomplex}(\kcomplex)=\sqrt{\kcomplex^{2n}+1}$ is analytic in $\CC^+\setminus\D^+$ and $\CC^-\setminus\D^-$.
Let $G(\kcomplex,t):=e^{-\wcomplex(\kcomplex)t}w(\kcomplex)\tilde{g}(\kcomplex,t)$.
Since $G(\kcomplex,t)$ is analytic in $\CC^+\setminus\D^+$, the integral of $G(\kcomplex,t)$ along the boundary $\C$ of the region $\CC^+\setminus\D^+$ is zero by \Cref{lemm:cauchy}, i.e.,
\begin{equation}\label{eq:cauchy-inegral}
		\begin{aligned}
				\int_{\C}e^{i\kcomplex x}G(\kcomplex,t) \diff \kcomplex = \left(\int_{-R}^{R} + \int_{\C_{R_2}} - \int_{\partial\D^+} + \int_{\C_{R_1}}\right)e^{i\kcomplex x}G(\kcomplex,t) \diff \kcomplex = 0.
		\end{aligned}
\end{equation}

Now we show that the integrals along $\C_{R_1}$ and $\C_{R_2}$ in \eqref{eq:cauchy-inegral} vanish as $R\to\infty$.
This follows from \Cref{lemm:jordan} with  $\text{Re}[\wcomplex(\kcomplex)]>0,\text{Re}[\bar{\wcomplex}(\kcomplex)]\geq0$, i.e., $\text{Re}[\sqrt{\kcomplex^{2n}+1}]>0$ in $\CC^+\setminus\D^+$.
This requirement is used to show that the 
$G(\kcomplex,t)$ decays exponentially as $|\kcomplex|\to\infty$ in this region.
The real part of the square root of any complex variable $a+b\im$ is given by $\sqrt{\frac{1}{2}(\sqrt{a^2+b^2}+a)}$. 
Consider $a=\text{Re}[\kcomplex^{2n}+1]$ and $b=\text{Im}[\kcomplex^{2n}+1]$. Then, the real part of $\sqrt{\kcomplex^{2n}+1}$ is always nonnegative and equals zero only if $b=\text{Im}[\kcomplex^{2n}+1]=\text{Im}[\kcomplex^{2n}]=0$. 
Using the definition $\kcomplex=|\kcomplex|e^{\im\theta}$, we have $\text{Im}[\kcomplex^{2n}]=|\kcomplex|^{2n}\sin(2n\theta)=0$ if $\theta=\frac{m\pi}{2n}, m=0,1,\ldots,4n-1$. Since $\theta\in(0,\frac{\pi}{2n})$ for $\kcomplex\in\CC^+\setminus\D^+$, $\text{Im}[\kcomplex^{2n}]\neq0$ and thus $\text{Re}[\wcomplex(\kcomplex)]=\text{Re}[\bar{\wcomplex}(\kcomplex)]>0$ in $\CC^+\setminus\D^+$ and $\CC^-\setminus\D^-$.

With $\text{Re}[\wcomplex(\kcomplex)]>0$, we show that $G(\kcomplex,t)\to0$ uniformly in $\C_{R_1}$ and $\C_{R_2}$ as $|\kcomplex|\to\infty$, i.e., $|G(\kcomplex,t)|\leq K_R$ in $\C_{R_1}$ and $\C_{R_2}$, where $K_R$ depends only on $R$ (not on argument of $\kcomplex$) and $K_R\to0$ as $R\to\infty$.
Note that ${g}_j(\kcomplex,t)$ contains the exponential term $e^{\wcomplex(\kcomplex)\tau}, 0<\tau<t$ for all $j=0,\ldots,n-1$. 
	Then, the leading exponential term in $e^{-\wcomplex(\kcomplex)t}\tilde{g}_j(\kcomplex,t)$ is $e^{-\wcomplex(\kcomplex)(t-\tau)}$.
	If $\text{Re}[\wcomplex(\kcomplex)]>0$ and $t-\tau>0$, $e^{-\wcomplex(\kcomplex)(t-\tau)}=e^{-\text{Re}[\wcomplex(\kcomplex)](t-\tau)}e^{-\im\text{Im}[\wcomplex(\kcomplex)](t-\tau)}$ decays exponentially and cancels the polynomial growth of $w(\kcomplex)$ and $c_j(\kcomplex)$ for all $j=0,\ldots,n-1$. 
	Similarly, $e^{-\wcomplex(\kcomplex)t}\underline{g}_j(\wcomplex(\kcomplex),\bar{\wcomplex}(\kcomplex),t)$ decays exponentially provided in addition that $\text{Re}[\bar{\wcomplex}(\kcomplex)]\geq0$.

Following \Cref{lemm:jordan}, the integrals along $\mathcal{C}_{R_1}$ and $\mathcal{C}_{R_2}$ in \eqref{eq:cauchy-inegral} vanish and we have
\begin{equation*}
	\begin{aligned}
		&\int_{-\infty}^{\infty}e^{\im \kcomplex x}G(\kcomplex,t)\frac{\diff \kcomplex}{2\pi} = \int_{\partial\D^+}e^{\im \kcomplex x}G(\kcomplex,t)\frac{\diff \kcomplex}{2\pi}.
	\end{aligned}
\end{equation*}

Similarly, let $H(\kcomplex,t):=e^{-\wcomplex(\kcomplex)t}\wcomplex(\kcomplex)\tilde{h}(\kcomplex,t)$. 
Since $H(\kcomplex,t)$ is analytic and $\text{Re}[\wcomplex(\kcomplex)]>0,\text{Re}[\bar{\wcomplex}(\kcomplex)]\geq0$ in $\CC^-\setminus\D^-$, applying \Cref{lemm:cauchy} and \Cref{lemm:jordan} in $\CC^-$ gives
\begin{equation*}
	\begin{aligned}
		&\int_{-\infty}^{\infty}e^{\im \kcomplex(x-L)}H(\kcomplex,t)\frac{\diff \kcomplex}{2\pi} = \int_{\partial\D^-}e^{\im \kcomplex(x-L)}H(\kcomplex,t)\frac{\diff \kcomplex}{2\pi}.
	\end{aligned}
\end{equation*}

\subsubsection{Elimination of unknown boundary values}
\begin{figure}[tb]
	\centering
	\includegraphics[width=0.35\linewidth]{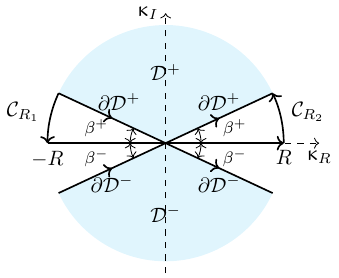}
	\caption{\sf Contour deformation from $\RR$ to $\partial \D^+$ and $\partial \D^-$}
	\label{fig:contour-deformation}
\end{figure}
\label{sec:general-procedures-elimination}
In the final step, we show how to eliminate unknown boundary values in \eqref{eq:u-2-integral-complex} for the even-order PDE \eqref{eq:even-order-pde}.

We start by noting that functions $\tilde{g}_j,\underline{g}_j, \tilde{h}_j,\underline{h}_j$ in \eqref{eq:global-relation-subsitute-u} depend on $\kcomplex$ only through $\wcomplex(\kcomplex)$ and $\bar{\wcomplex}(\kcomplex)$ or equivalently through $w_{\text{Re}}(\kcomplex)$.
Therefore, these functions are \emph{invariant} under any transformation $\kcomplex\to\lambda(\kcomplex)$ that leaves $w_{\text{Re}}(\kcomplex)$ invariant, i.e., $w_{\text{Re}}(\kcomplex)=w_{\text{Re}}(\lambda(\kcomplex))$.
Note that $\hat{p}(\kcomplex)$ is also invariant under such transformations.
There are $n$ distinct solutions to $w_{\text{Re}}(\kcomplex)=w_{\text{Re}}(\lambda(\kcomplex))$, given by $\lambda_l(\kcomplex)=\kcomplex \exp[\im\frac{ 2\pi l}{n}],l=0,\ldots,n-1$.
By replacing $\kcomplex$ in the global relation \eqref{eq:global-relation-subsitute-u} with $\lambda_l(\kcomplex)$, we obtain a system of $n$ equations indexed by $l=0,\ldots,n-1$,
\begin{equation}\label{eq:global-relation-invariant-lambda}
	\begin{aligned}
		&e^{\wcomplex(\kcomplex)t}\statehat(\lambda_l(\kcomplex),t) = \statehat_0(\lambda_l(\kcomplex)) - \sum_{j=0}^{n-1}c_j(\lambda_l(\kcomplex))\big[\check{g}_j(\kcomplex,t)- e^{-\im \lambda_l(\kcomplex) L} \check{h}_j(\kcomplex,t)\big].
	\end{aligned}
\end{equation}
Recall that $n$ among the $2n$ functions $\{\check{g}_j,\check{h}_j\}_{j=0}^{n-1}$ in \eqref{eq:check-g-h} can be computed from the $n$ given boundary conditions.
Let $p,q$ denote the indices of unknown boundary values at $x=0$ and $x=L$ respectively, i.e.,  $g_p$ and $h_q$ are unknown, where $p$ takes $n-N$ integer values and $q$ takes $n$ integer values, see \Cref{assumption:well-posed}.
Then, \eqref{eq:global-relation-invariant-lambda} is a system of $n$ equations for $n$ unknown functions $\check{g}_p,\check{h}_q$.
The existence and uniqueness of the solution $\check{g}_p$ and $\check{h}_q$ to \eqref{eq:global-relation-invariant-lambda} for $\kcomplex\in\CC^+\setminus\D^+$ and $\CC^-\setminus\D^-$ is guaranteed by \Cref{assumption:well-posed} \cite[pg. 65]{fokas2008unified}.   
Solving \eqref{eq:global-relation-invariant-lambda}, $\check{g}_p$ and $\check{h}_q$ are represented in terms of the given initial and boundary conditions, and the remaining unknown involving $\statehat(\lambda_l(\kcomplex),t),l=0\ldots,n-1$.  
Substituting the solution for $\check{g}_p$ and $\check{h}_q$ into \eqref{eq:u-2-integral-complex}, we shall show that the integrals involving the remaining unknown $\statehat(\lambda_l(\kcomplex),t)$ vanish.
This follows from the procedure in \cite[Proposition 1.2]{fokas2008unified}.
For brevity, we will only illustrate this for the reaction-diffusion equation in \Cref{sec:elimination-unknown-boundary-reaciton}.

\begin{remark}\label{re:elemination-unknown-boundary}
	To obtain a system of $n$ equations \eqref{eq:global-relation-invariant-lambda} for general PDE \eqref{eq:evolution-equation}, we require $n$ distinct solutions $\lambda_l(\kcomplex),l=0,\ldots,n-1$ such that $w_{\text{Re}}(\kcomplex)=w_{\text{Re}}(\lambda_l(\kcomplex))$ and $w_{\text{Im}}(\kcomplex)=w_{\text{Im}}(\lambda_l(\kcomplex))$. 
	This is guaranteed if $w_{\text{Im}}$ or $w_{\text{Re}}$ is zero.
	Otherwise, it is necessary for both polynomials $w_{\text{Re}}(\kcomplex)$ and $w_{\text{Im}}(\kcomplex)$ to have degree $n$.
	For example, if $w_{\text{Re}}(\kcomplex) = \kcomplex^2$ and $w_{\text{Im}}(\kcomplex)=\kcomplex$ for $n=2$, then $\lambda_0(\kcomplex)=\kcomplex$ is the only solution.
\end{remark}
	
\subsubsection{Vanishing integral}
Following the above three steps, we can derive an expression for $u^*_2(x,t)$ and similarly for $u^*_1(x,t)$ only in terms of the given initial and boundary conditions.
Now, we analyze $u^*_3(x,t)$ in \eqref{eq:u-opt-original}.
Expanding $v(k,t)$ from \eqref{eq:boundary-value-v}, we obtain
\begin{align*}
	u^*_3(x,t) &= \int_{-\infty}^{\infty}\sum_{j=0}^{n-1}e^{\im k x}c_j(k)\big[g_j(t) - e^{-\im k L}h_j(t)\big]\frac{\diff k}{2\pi} \\
	& = \sum_{j=0}^{n-1}g_j(t)\!\int_{-\infty}^{\infty}e^{\im k x}c_j(k) \frac{\diff k}{2\pi}\! -\!h_j(t)\!\int_{-\infty}^{\infty}e^{\im k (x-L)}c_j(k)\frac{\diff k}{2\pi}.
\end{align*}
Recall that $c_j(k),j=0,\ldots,n-1$ obtained from \eqref{eq:X-function-expansion} are polynomials of $k$ and thus can be expressed as $c_j(k)=\sum_{l=0}^{n_j} \tilde{c}_{jl}(\im k)^l$ with coefficients $\tilde{c}_{jl}\in\CC,l=0,\ldots,n_j$. Then, we find
\begin{equation}\label{eq:u-3-dirac-delta}
	u^*_3(x,t) = \sum_{j=0}^{n-1}\sum_{l=0}^{n_j}\tilde{c}_{jl}\left[ \delta^{(j)}(x)g_j(t)-\delta^{(j)}(x-L)h_j(t)\right],
\end{equation}
where $\delta^{(j)}(x)$ is the $j$-th derivative of the Dirac delta function.
For the open interval $0<x<L$ in the domain $\Omega$, the delta function and its derivatives equal zero; thus, $u^*_3(x,t)=0$. 
\begin{figure}[b]
	\centering
	\includegraphics[width=0.3\linewidth]{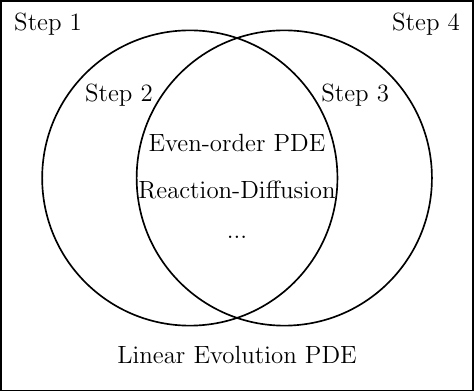}
	\caption{\sf Classes of PDEs satisfying the requirements of the general procedure}
	\label{fig:venn-diagram}
\end{figure}

We have presented all the needed techniques to remove the dependence of optimal control \eqref{eq:u-opt-original} on unknown boundary values.
The requirements of the general procedure for the class of even-order PDEs \eqref{eq:even-order-pde} are summarized in \Cref{table:sample-procedure}.
The relations between different requirements are illustrated in \Cref{fig:venn-diagram}.
In the following subsection, we apply these techniques to the reaction-diffusion equation that does not fit into \eqref{eq:even-order-pde}.

\subsection{The reaction-diffusion equation}
\label{sec:optimal-control-reaciton-diffusion}
Following \Cref{example:problem-formulation}, $w_{\text{Re}}(\kcomplex) = \kcomplex^2+c$ and $w_{\text{Im}}(\kcomplex) = 0$. 
Then, $\wcomplex(\kcomplex)=\bar{\wcomplex}(\kcomplex)=\sqrt{(\kcomplex^2+c)^2+1}$. Since $\wcomplex(\kcomplex)=\bar{\wcomplex}(\kcomplex)$, we will omit the arguments of $\bar{\wcomplex}(\kcomplex)$ in $\underline{g}_j,\underline{h}_j,j=0,\ldots,n-1$, see \eqref{eq:global-relation-subsitute-u}. Using the procedure in \Cref{sec:general-procedures}, we obtain the following result.
\begin{table}[tb]
	\small
		\centering
		\caption{\sf Requirements for the general procedure}
		\label{table:sample-procedure}
		\begin{tabular}{cc}
			\hline
			Step & Requirement \\
			\hline
			1 & None\\
			2 & $\text{Re}[\wcomplex(\kcomplex)]>0,\text{Re}[\bar{\wcomplex}(\kcomplex)]\geq0, \kcomplex\in\mathcal{D}$ \\
			& for some region $\mathcal{D}\subset\CC$ such that $ \RR\subset\mathcal{D}$ \\
			3 & There exist $n$ distinct solutions $\lambda_l(\kcomplex),l=0,\ldots,n-1$
			\\ & to $w_{\text{Re}}(\kcomplex)=w_{\text{Re}}(\lambda_l(\kcomplex)), w_{\text{Im}}(\kcomplex)=w_{\text{Im}}(\lambda_l(\kcomplex))$ \\
			4 & None \\
			\hline
		\end{tabular}
	\end{table}

\begin{theorem}\label{theo:u-opt-effective-heat}
	Consider the reaction-diffusion equation in \eqref{eq:heat} with given Dirichlet boundary conditions $g_0(t),h_0(t)$. Optimal control $u^*(x,t)$ in the domain $\Omega$ with $T\to\infty$ is given by
	\begin{equation}\label{eq:u-opt-effective}
		\begin{aligned}
			&u^*(x,t) = -\int_{-\infty}^{\infty}e^{\im \kcomplex x-\wcomplex(\kcomplex)t}\, \hat{p}(\kcomplex)\statehat_0(\kcomplex)\dfrac{\diff \kcomplex}{2\pi} + \int_{\partial\D^+}\dfrac{e^{-\wcomplex(\kcomplex)t}\hat{p}(\kcomplex)}{\Delta(\kcomplex)}U_1(\kcomplex,x,t)\dfrac{\diff \kcomplex}{2\pi}- \int_{\partial\D^+}\dfrac{U_2(\kcomplex,x,t)}{\Delta(\kcomplex)} \frac{\diff \kcomplex}{2\pi},
		\end{aligned}
	\end{equation}
	where $\Delta(\kcomplex)=e^{\im \kcomplex L} - e^{-\im \kcomplex L}, \partial \D^+=\{\kcomplex\in\CC^+:\kcomplex=\abs{\kcomplex}e^{\im \theta}, \theta = \frac{1}{4}\arctan(1/c) \text{ or }-\frac{1}{4}\arctan(1/c)+\pi\}$, and
	\begin{equation*}
		\begin{split}		
			U_1(\kcomplex,x,t):=&\, 2\sin(\kcomplex x)\big[\im e^{\im \kcomplex L}\statehat_0(\kcomplex) - 2\kcomplex(\tilde{h}_0(\wcomplex(\kcomplex),t) - \hat{p}(\kcomplex)\underline{h}_0(\wcomplex(\kcomplex),t))\big]\\
			&+2\sin(\kcomplex(L-x))\big[\im \statehat_0(-\kcomplex) - 2\kcomplex(\tilde{g}_0(\wcomplex(\kcomplex),t)-\hat{p}(\kcomplex)\underline{g}_0(\wcomplex(\kcomplex),t))\big], 	
		\end{split}
	\end{equation*}
	\begin{equation*}
		\begin{split}
			U_2(\kcomplex,x,t):=&2\sin(\kcomplex x) \big[-2\kcomplex(h_0(t) - \hat{p}(\kcomplex)\underline{h}_0(\wcomplex(\kcomplex),t))\big] + 2\sin(\kcomplex(L-x))\big[-2\kcomplex(g_0(t) - \hat{p}(\kcomplex)\underline{g}_0(\wcomplex(\kcomplex),t))\big].
		\end{split}
	\end{equation*}
\end{theorem}

\begin{remark}
	The expression \eqref{eq:u-opt-effective} reduces to \eqref{eq:optimal-control-heat-integral} for the heat equation with homogeneous Dirichlet boundary conditions.
\end{remark}

\begin{proof}
	Following \eqref{eq:u-opt-original}, $u^*(x,t)$ is equal to the time derivative of
	\begin{equation}\label{eq:state-expression}
		\begin{aligned}
			\state(x,t)&=\int_{-\infty}^{\infty}e^{\im \kcomplex x}\statehat(\kcomplex,t)\frac{\diff \kcomplex}{2\pi}=\int_{-\infty}^{\infty}e^{\im \kcomplex x-\wcomplex(\kcomplex)t}\, \statehat_0(\kcomplex)\dfrac{\diff \kcomplex}{2\pi}- \int_{\partial\D^+}\dfrac{e^{-\wcomplex(\kcomplex)t}}{\Delta(\kcomplex)}U_1(\kcomplex,x,t) \dfrac{\diff \kcomplex}{2\pi},
		\end{aligned}
	\end{equation}
	plus
	\begin{equation}\label{eq:state-k-2-after-deform}
		\begin{aligned}
			\int_{-\infty}^{\infty}e^{\im \kcomplex x}w(\kcomplex)\statehat(\kcomplex,t)\frac{\diff \kcomplex}{2\pi}&=\int_{-\infty}^{\infty}e^{\im \kcomplex x-\wcomplex(\kcomplex)t}\, w(\kcomplex)\statehat_0(\kcomplex)\dfrac{\diff \kcomplex}{2\pi} - \int_{\partial\D^+}\dfrac{e^{-\wcomplex(\kcomplex)t}w(\kcomplex)}{\Delta(\kcomplex)}U_1(\kcomplex,x,t) \dfrac{\diff \kcomplex}{2\pi},
		\end{aligned}
	\end{equation}
	and plus the following integral that vanishes,
	\begin{equation}\label{eq:inverse-transform-boundary}
		\int_{-\infty}^{\infty}e^{\im k x}v(k,t)\frac{\diff k}{2\pi} = 0, \quad 0<x<L.
	\end{equation}
% \begin{remark}
% 	\eqref{eq:D-N-map} is also known as the Dirichlet-to-Neumann map in boundary value problems for PDEs. Conversely,
% 	if the Neumann boundary values are given, we can use a procedure similar to the proof and derive the expression for Dirichlet boundary values that depend on the given initial and Neumann boundary values.
% \end{remark}

	We will only prove \eqref{eq:state-k-2-after-deform} and divide the proof into three steps, analogously to \Cref{sec:general-procedures}. 
	%The proof of each step can be found in the extended version.
	The proof of \eqref{eq:state-expression} follows similarly without the extra $w(k)$ term. 
	Equation \eqref{eq:inverse-transform-boundary} follows from \eqref{eq:u-3-dirac-delta}.
	
\subsubsection{The integral representation}
 Applying \eqref{eq:u-2-integral-real} to the reaction-diffusion equation, we obtain
 \begin{equation}\label{eq:state-before-deform}
 	\begin{aligned}
 		u^*_2(x,t) &= \int_{-\infty}^{\infty}e^{\im k x-\wcomplex(k)t}w(k)\statehat_0(k)\frac{\diff k}{2\pi} -  \int_{-\infty}^{\infty}e^{\im k x-\wcomplex(k)t}w(k)\tilde{g}(k,t)\frac{\diff k}{2\pi}+ \int_{-\infty}^{\infty}e^{\im k(x-L)-\wcomplex(k)t}w(k)\tilde{h}(k,t) \frac{\diff k}{2\pi},\\
 	\end{aligned}
 \end{equation}
 where 
 \begin{align*}
 	\tilde{g}(k,t)=&\im k(\tilde{g}_0(\wcomplex(k),t) - \hat{p}(k)\underline{g}_0(\wcomplex(k),t)) + \tilde{g}_1(\wcomplex(k),t) - \hat{p}(k)\underline{g}_1(\wcomplex(k),t),\\
 	\tilde{h}(k,t)=&\im k(\tilde{h}_0(\wcomplex(k),t) - \hat{p}(k)\underline{h}_0(\wcomplex(k),t)) + \tilde{h}_1(\wcomplex(k),t) - \hat{p}(k)\underline{g}_1(\wcomplex(k),t).
 \end{align*}

\subsubsection{Contour deformation}
Following \Cref{sec:contour-deformation}, we first examine the analyticity of $\wcomplex(\kcomplex)=\sqrt{(\kcomplex^2+c)^2+1}$.
Since $\wcomplex(\kcomplex)$ is a square root function in $\CC$, we need to find the branch points and the corresponding branch cuts.
\begin{figure}[tb]
	\centering
	\includegraphics[width=0.35\linewidth]{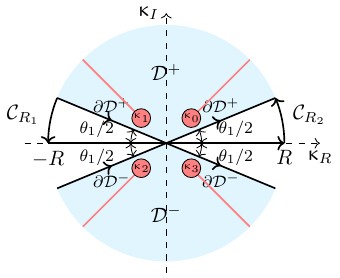}
	\caption{\sf Contour deformation for \eqref{eq:state-before-deform}. The red lines are branch cuts of $\wcomplex(\kcomplex)$.}
	\label{fig:contour-deformation-heat}
\end{figure}
The branch points include infinity and points where $\wcomplex(\kcomplex)=0$, i.e., $\kcomplex^2+c=\pm\im$.
Using the definition $\kcomplex=\abs{\kcomplex}e^{\im\theta}$, these points are $\kcomplex_m=(c^2+1)^{1/4}e^{\im\theta_m},m=1,\ldots,4$, where $\theta_1=\frac{1}{2}\arctan(1/c),\theta_2=-\frac{1}{2}\arctan(1/c),\theta_3=-\frac{1}{2}\arctan(1/c)+\pi,\theta_4=\frac{1}{2}\arctan(1/c)+\pi$. Then, the branch cuts can be chosen as the four rays connecting $\kcomplex_m$ and infinity with arguments $\theta_m,m=1,\ldots,4$, see the red lines in \Cref{fig:contour-deformation-heat}.
The function $\wcomplex(\kcomplex)$ is analytic in $\CC$, excluding the branch cuts. 
Therefore, we must avoid the four branch cuts for contour deformation.

Then, we consider the integral with $\tilde{g}(\kcomplex,t)$ in \eqref{eq:state-before-deform} and rewrite the integrand as $e^{i\kcomplex x}G(\kcomplex,t)$ with $G(\kcomplex,t) := \ e^{-\wcomplex(\kcomplex)t}w(\kcomplex)\tilde{g}(\kcomplex,t)$
Following the discussion in \Cref{sec:contour-deformation}, $G(\kcomplex,t)$ decays exponentially provided that $\text{Re}[\wcomplex(\kcomplex)] >0$. 
The real part of the square root of any complex variable $a+b\im$ is given by $\sqrt{\frac{1}{2}(\sqrt{a^2+b^2}+a)}$. 
Therefore, the real part is always nonnegative and equals zero only if $b=0$. Using the definition $\kcomplex=\kcomplex_R+\im \kcomplex_I$, $\text{Im}[(\kcomplex^2+c)^2 + 1]=0$ is equivalent to $\kcomplex=0$ or $\kcomplex_I=\pm\sqrt{\kcomplex_R^2+c}$. If $\kcomplex=0$, $\text{Re}[\wcomplex(\kcomplex)] = \sqrt{c^2+1}>0$. If $\kcomplex_I=\pm\sqrt{\kcomplex_R^2+c}\sim \pm \kcomplex_R$, as $\kcomplex\to\infty$, the curve $\text{Im}[(\kcomplex^2+c)^2 + 1]=0$ approaches the four rays that connect the origin and infinity defined by $\kcomplex_I=\pm \kcomplex_R$. The arguments of these four rays are $\pm\pi/4$ and 
$\pm3\pi/4$. Since $\theta_1=\frac{1}{2}\arctan(1/c)\leq\pi/4$ with equality holds only when $c=0$, $ \text{Re}[\wcomplex(\kcomplex)] >0$ as $\kcomplex=\abs{\kcomplex}e^{\im\theta}\to\infty$ when  $\theta_2<\theta<\theta_1$ and $\theta_3<\theta<\theta_4$.

Without loss of generality, let $\D^+=\{\kcomplex\in\CC^+:\kcomplex=|\kcomplex|e^{\im\theta}, \theta_1/2<\theta<-\theta_1/2+\pi\}$, $\partial \D^+=\{\kcomplex\in\CC^+:\kcomplex=\abs{\kcomplex}e^{\im \theta}, \theta = \theta_1/2 \text{ or } -\theta_1/2+\pi\}$ with the direction from left to right.
The integrand $e^{\im \kcomplex x}G(\kcomplex,t)$ is analytic in $\CC^+\setminus\D^+$ since the four branch cuts are avoided and the integrand is bounded, see \Cref{fig:contour-deformation-heat}.
Consider a contour $\C = [-R,R]\cup\C_{R_2}\cup\partial\D^+\cup\C_{R_1}$ shown in \Cref{fig:contour-deformation-heat}.
From \Cref{lemm:cauchy} and the analyticity of the integrand $e^{i \kcomplex x}G(\kcomplex,t)$ in the domain enclosed by $\C$, we have \eqref{eq:cauchy-inegral}.
Taking the limit $R\to\infty$, the integrals over $\C_{R_1}$ and $\C_{R_2}$ vanish according to \Cref{lemm:jordan}. Repeating the analysis for the integral with $\tilde{h}(\kcomplex,t)$ in \eqref{eq:state-before-deform} and replacing $\partial\D^+$ by $\partial \D^-=\{\kcomplex\in\CC^-:\kcomplex=\abs{\kcomplex}e^{\im \theta}, \theta = -\theta_1/2\text{ or } \theta_1/2+\pi\}$ with the direction from left to right, we obtain \eqref{eq:u-2-integral-complex}. Using the transformation $\kcomplex\to-\kcomplex$ for the integral along $\partial\D^-$, we find
\begin{equation}\label{eq:state-after-deform-D+}
	\begin{aligned}
		&u^*_2(x,t)=\int_{-\infty}^{\infty}e^{\im \kcomplex x-\wcomplex(\kcomplex)t}\, w(\kcomplex)\statehat_0(\kcomplex)\dfrac{\diff \kcomplex}{2\pi} -\int_{\partial\D^+} e^{-\wcomplex(\kcomplex)t}w(\kcomplex)\left[e^{\im \kcomplex x}\tilde{g}(\kcomplex,t)+e^{\im \kcomplex(L-x) }\tilde{h}(\kcomplex,t) \right]\dfrac{\diff \kcomplex}{2\pi}.
	\end{aligned}
\end{equation}
\begin{remark}
	The transformation from $\partial\D^-$ to $\partial\D^+$ is not necessary and is used here for convenience later to eliminate unknown boundary values.
\end{remark}
\subsubsection{Elimination of unknown boundary values}
\label{sec:elimination-unknown-boundary-reaciton}
In this step, we show that the unknown boundary values can be eliminated in \eqref{eq:state-after-deform-D+}.
We use the reaction-diffusion equation to illustrate the details omitted in \Cref{sec:general-procedures-elimination}. 
Note that $\wcomplex(-\kcomplex)=\wcomplex(\kcomplex)$ and $\hat{p}(-\kcomplex)=\hat{p}(\kcomplex)$. Then, we obtain two global relations from \eqref{eq:global-relation-invariant-lambda} with $\lambda_0(\kcomplex)=\kcomplex$ and $\lambda_1(\kcomplex)=-\kcomplex$.
In the case that $g_0(t), h_0(t)$ are given, from \eqref{eq:global-relation-invariant-lambda} we have
\begin{equation}\label{eq:global-relation-two}
	\begin{aligned}
		e^{\wcomplex(\kcomplex)t}\statehat(\kcomplex,t)&=G_0(\kcomplex,t) - (\tilde{g}_1 - \hat{p}(\kcomplex)\underline{g}_1) + e^{-\im \kcomplex L}(\tilde{h}_1 -\hat{p}(\kcomplex) \underline{h}_1), \\ 
		e^{\wcomplex(\kcomplex)t}\statehat(-\kcomplex,t)&=G_0(-\kcomplex,t) - (\tilde{g}_1 - \hat{p}(\kcomplex)\underline{g}_1)+ e^{\im \kcomplex L}(\tilde{h}_1 -\hat{p}(\kcomplex) \underline{h}_1),
	\end{aligned}
\end{equation}
where $G_0(\kcomplex,t)$ containing given initial and boundary conditions is defined by
\begin{equation*}
	G_0(\kcomplex,t) := \statehat_0(\kcomplex) - \im \kcomplex(\tilde{g}_0 -\hat{p}(\kcomplex) \underline{g}_0) + \im \kcomplex e^{-\im \kcomplex L}(\tilde{h}_0 -\hat{p}(\kcomplex) \underline{h}_0).
\end{equation*}
Solving \eqref{eq:global-relation-two} for terms containing the unknown boundary values, we find
\begin{equation*}
	\begin{aligned}
		\tilde{g}_1 - \hat{p}(\kcomplex)\underline{g}_1&= \frac{1}{\Delta(\kcomplex)}\left[e^{\im \kcomplex L}G_0(\kcomplex,t)-e^{-\im \kcomplex L}G_0(-\kcomplex,t)\right] + e^{\wcomplex(\kcomplex)t}\acute{g}_1,\\
		\tilde{h}_1 - \hat{p}(\kcomplex)\underline{h}_1&= \frac{1}{\Delta(\kcomplex)}\left[G_0(\kcomplex,t) - G_0(-\kcomplex,t)\right] + e^{\wcomplex(\kcomplex)t}\acute{h}_1,
	\end{aligned}
\end{equation*}
where
	$\Delta(\kcomplex)= e^{\im \kcomplex L} - e^{-\im \kcomplex L},$
and $\acute{g}_1,\acute{h}_1$ contain the unknown terms $\statehat(\kcomplex,t)$ and $\statehat(-\kcomplex,t)$ as follows:
\begin{align}
			\acute{g}_1 &= -\frac{1}{\Delta(\kcomplex)}\left[e^{\im \kcomplex L}\statehat(\kcomplex,t) - e^{-\im \kcomplex L}\statehat(-\kcomplex,t)\right], \label{eq:state-transform-unknown}\\
		\acute{h}_1 &= -\frac{1}{\Delta(\kcomplex)}\left[\statehat(\kcomplex,t) - \statehat(-\kcomplex,t)\right].
\end{align}
Inserting the expressions for $\tilde{g}_1-\hat{p}(\kcomplex)\underline{g}_1$ and $\tilde{h}_1-\hat{p}(\kcomplex)\underline{h}_1$ in \eqref{eq:state-after-deform-D+} and simplifying, we find \eqref{eq:state-k-2-after-deform} plus the following term that vanishes,
\begin{equation}\label{eq:vanish-integral}
	\int_{\partial\D^+}\left[e^{\im \kcomplex x}w(\kcomplex)\acute{g}_1(\kcomplex,t) - e^{\im \kcomplex(L-x)}w(\kcomplex)\acute{h}_1(\kcomplex,t)\right]\dfrac{\diff \kcomplex}{2\pi}.
\end{equation}
\begin{figure}[tb]
	\centering
	\includegraphics[width=0.35\linewidth]{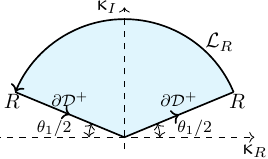}
	\caption{\sf Contour deformation for \eqref{eq:fokas-vanishing-integral}}
	\label{fig:contour-deformation-D+}
\end{figure}

To show that the integral \eqref{eq:vanish-integral} vanishes, we analyze the behavior of the integrand as $\kcomplex\to \infty$ to use \Cref{lemm:jordan}. 
Expanding $\Delta(\kcomplex)$ in \eqref{eq:state-transform-unknown}, we have $\acute{g}_1 = (-e^{2\im \kcomplex L}\statehat(\kcomplex,t)+\statehat(-\kcomplex,t))/(e^{2\im \kcomplex L}-1)$.
Using the definition $\kcomplex=\kcomplex_R+\im\kcomplex_I$, the exponent $e^{\im \kcomplex\lambda}=e^{-\kcomplex_I\lambda}e^{\im \kcomplex_R\lambda}$ decays exponentially for $\kcomplex\in\CC^+$ with $\lambda>0$. 
Therefore, the denominator in $\acute{g}_1$ satisfies $\lim_{\kcomplex\to\infty} e^{2\im \kcomplex L}-1 = -1$ since $L >0$. 
For the nominator in $\acute{g}_1$, the first term $e^{2\im \kcomplex L}\statehat(\kcomplex,t) = e^{\im \kcomplex L}\int_{0}^{L}e^{\im \kcomplex (L-x)}\state(x,t)\diff x$ decays exponentially since $L>0,L-x>0$.
The second term $\statehat(-\kcomplex,t)$ also decays exponentially since it contains the exponent $e^{\im \kcomplex x}$ with $x>0$.
Therefore, $\acute{g}_1$ decays exponentially for $\kcomplex\in\CC^+$. 

Consider a closed curve $\mathcal{L} = \mathcal{L}_{\partial \D^+}\cup\mathcal{L}_{R}$, where $\mathcal{L}_{\partial\D^+} = \partial\D^+\cap\{\kcomplex:|\kcomplex|<R\}$ and $\mathcal{L}_R=\{\kcomplex\in\D^+:|\kcomplex|=R\}$, see \Cref{fig:contour-deformation-D+}. 
Note that $\acute{g}$ in \eqref{eq:state-transform-unknown} is analytic in $\CC^+$.
From \Cref{lemm:cauchy}, we have\footnote{There exists a pole $\kcomplex=0$ in $\acute{g}(\kcomplex,t)$ due to the fraction $1/\Delta(\kcomplex)$. 
Since $\lim_{\kcomplex\to0}w(\kcomplex)\acute{g}_1(\kcomplex,t)=0$, this pole is a removable singularity for the integrand, i.e., the Taylor series expansion that approximates the integrand near $\kcomplex=0$ is well-defined for all values of $\kcomplex$, including $\kcomplex=0$. The removable singularity can be ignored in the integral.}
\begin{equation}\label{eq:fokas-vanishing-integral}
	\begin{aligned}
		\int_{\mathcal{L}} e^{\im \kcomplex x}w(\kcomplex)\acute{g}_1(\kcomplex,t) \diff \kcomplex
		= \int_{\mathcal{L}_{\partial\D^+}}\!\!\! e^{\im \kcomplex x}w(\kcomplex)\acute{g}_1(\kcomplex,t) \diff \kcomplex  + \int_{\mathcal{L}_R}\!\!e^{\im \kcomplex x}w(\kcomplex)\acute{g}_1(\kcomplex,t) \diff \kcomplex = 0.
	\end{aligned}
\end{equation}
The integral along $\mathcal{L}_R$ vanishes as $R\to\infty$ according to \Cref{lemm:jordan} and the fact that $w(\kcomplex)\acute{g}_1\to0$ as $R\to\infty$ since the exponential decay of $\acute{g}_1$ cancels the polynomial growth of $w(\kcomplex)$. 
From \eqref{eq:fokas-vanishing-integral}, the integral along $\mathcal{L}_{\partial\D^+}$ also goes to zero as $R\to\infty$. 
Since $\mathcal{L}_{\partial\D^+}$ converges to $\partial\D^+$ as $R\to\infty$, the contribution of $\acute{g}_1$ in the integral \eqref{eq:vanish-integral} vanishes. 
Similarly, the contribution of $\acute{h}_1$ in the integral \eqref{eq:vanish-integral} also vanishes. 
Therefore, the value of the integral \eqref{eq:vanish-integral} is zero. 
\end{proof}

\section{Optimal control in feedback form}\label{sec:feedback-form}

The integral representation \eqref{eq:u-opt-effective} can be readily computed since it only depends on the given initial and boundary conditions. However, it is unclear from the expression how optimal control depends on the state variable $\state(x,t)$. In this section, we rewrite \eqref{eq:u-opt-effective} in a state-feedback convolution form. 
% Then, we analyze the structure of the resulting feedback kernel function.

For the unbounded domain $-\infty<x<\infty$, \eqref{eq:u-opt-transform} reduces to $\hat{u}^*(k,t) = -\hat{p}(k)\statehat(k,t)$. By the convolution theorem for the Fourier transform, 
$$u^*(x,t)= -\int_{-\infty}^{\infty}e^{\im k x} \hat{p}(k)\statehat(k,t) = -\int_{-\infty}^{\infty}K(x-\xi)\state(\xi)\diff \xi,$$ 
where $K$ is the inverse Fourier transform of $\hat{p}(k)$.
The kernel $K$ reveals the feedback structure of optimal control since it determines the dependence of $u^*(x,t)$ on the state $\state(\xi,t)$ at different locations $\xi$.
For the bounded domain $0<x<L$, integrals in optimal control \eqref{eq:u-opt-effective} are not on the real line, and the convolution theorem does not apply directly.
It turns out that \eqref{eq:u-opt-effective} can be rewritten as a Fourier series to which the convolution theorem applies. The following theorem illustrates the result.

\begin{theorem}
	The integral representation \eqref{eq:u-opt-effective} is equivalent to the following feedback form:
	\begin{equation}\label{eq:u-opt-convolution}
		\begin{aligned}
			u^*_{\text{conv}}(x,t) = \int_{0}^{L}\left(\Gamma(x,\xi) - \Gamma(x,-\xi)\right)\state(\xi,t)\frac{\diff \xi}{2L}+ \sum_{m=1}^{\infty}\sin(k_m x)\underline{b}_m(t),
		\end{aligned}
	\end{equation}
	where $k_m = \pi m/L,m=1,2,\ldots,$
	\begin{equation}\label{eq:Gamma-convolution}
		\Gamma(x,\xi) = \sum_{m=-\infty}^{\infty}\hat{p}(m)e^{\im m(x-\xi)},
	\end{equation}
	$\hat{p}(m) = - (m^2+c) + \sqrt{(m^2+c)^2 + 1}$, and
	\begin{equation}\label{eq:underline-bn}
		\begin{aligned}
			\underline{b}_m(t) &= 2k_m\big[(-1)^m(h_0(t)-\hat{p}(k_m)\underline{h}_0(\wcomplex(k_m),t))  -(g_0(t)-\hat{p}(k_m)\underline{g}_0(\wcomplex(k_m),t))\big].
		\end{aligned}
	\end{equation}
\end{theorem}	

\begin{remark}\label{re:structure}
	The kernel function $\Gamma(x,\xi)-\Gamma(x,-\xi)$ determines the feedback structure of optimal control \eqref{eq:u-opt-convolution}, in particular the dependence on the state $\state(\xi,t)$ at different locations $\xi$. The kernel function is in the same form as the case of homogeneous boundary conditions discovered in \cite{epperlein2016spatially}. The only difference is the additive sine series term in \eqref{eq:u-opt-convolution}. 
    %The sine series vanishes when the boundary conditions are zero.

	In \cite{epperlein2016spatially}, it is proved that the kernel function exhibits a structure called \emph{Toeplitz plus Hankel}, where $\Gamma(x,\xi)$ has a diagonal structure, and $\Gamma(x,-\xi)$ has an antidiagonal structure. These properties will be illustrated numerically in \Cref{sec:numerical-experiment}.
\end{remark}
\begin{remark}
	Recall that \Cref{theo} provides optimal control of ODEs with exogenous inputs.
	The expression in \eqref{eq:u-opt-convolution} when $c=0$ can also be obtained by applying \Cref{theo} to ODE \eqref{eq:ode-separation-m} and rewriting the resulting sine series as a Fourier series in exponential form. 
	Recall that the derivation of \eqref{eq:ode-separation-m} requires an orthonormal basis for $\mathcal{L}^2((0,L),\RR)$ that varies for different PDEs.
    This is to be contrasted with a unified, i.e., independent of basis, approach to derive \eqref{eq:u-opt-convolution} as described in the proof of Theorem~\ref{theo:u-opt-effective-heat}. This approach is in the same spirit as the unified transform approach at the core of Section~\ref{sec:integral-representation-complex}, but differs in the specifics of contour deformation and related analysis. 
\end{remark}

%\begin{remark}
%	Rewriting the complex integrals in \eqref{eq:u-opt-convolution} as series relies on the fact that the integrands
%\end{remark}
\begin{figure}[tb]
	\centering
	\includegraphics[width=0.45\linewidth]{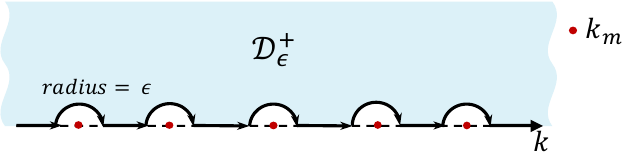}
	\caption{\sf Contour deformation for \eqref{eq:u-opt-effective-D-epsilon}}
	\label{fig:contour-deformation-D+-epsilon}
\end{figure}
\begin{proof}
We first reformulate the integral representation \eqref{eq:u-opt-effective} into an equivalent series representation. 
We start by noting that for the integrals along $\partial\D^+$ in \eqref{eq:u-opt-effective}, the denominator $\Delta(\kcomplex)=e^{\im \kcomplex L} - e^{-\im \kcomplex L} = 2\im\sin(\kcomplex L)$ has simple roots\footnote{A simple root of $\Delta(\kcomplex)$ at $\kcomplex=k_m$ means that $\Delta(k_m)=0$ but its first derivative is not equal to zero at $k_m$.} in $\RR$ at $k_m = \pi m/L,m\in\ZZ$.
The idea is to deform the complex contour $\partial\D^+$ back to the real line with small loops of radius $\epsilon\to0$ around these roots, denoted by $\partial\D_\epsilon^+$, which is the boundary of $\D_\epsilon^+$ shown in \Cref{fig:contour-deformation-D+-epsilon}.
Applying \Cref{lemm:cauchy,lemm:jordan} and following a similar analysis as in the Proof of \Cref{theo:u-opt-effective-heat}, the integrals along $\partial\D^+$ can be deformed to $\partial\D_\epsilon^+$.
We will restrict our attention to $\partial\D_\epsilon^+$ and \eqref{eq:u-opt-effective} can be rewritten as
\begin{equation}\label{eq:u-opt-effective-D-epsilon}
	\begin{aligned}
		u^*(x,t) = -\int_{-\infty}^{\infty}e^{\im \kcomplex x - \wcomplex(\kcomplex)t} \hat{p}(\kcomplex)\statehat_0(\kcomplex) \frac{\diff \kcomplex}{2\pi}+ \int_{\partial\D_\epsilon^+}f_1(x,t,\kcomplex)\frac{\diff \kcomplex}{2\pi} - \int_{\partial\D_\epsilon^+}f_2(x,t,\kcomplex)\frac{\diff \kcomplex}{2\pi},
	\end{aligned}
\end{equation}
where $f_1(x,t,\kcomplex)$ and $f_2(x,t,\kcomplex)$ are the integrands for the corresponding integrals in \eqref{eq:u-opt-effective}.
We will focus on the first two integrals in \eqref{eq:u-opt-effective-D-epsilon}. 
The third one can be evaluated similarly.
By \Cref{lemm:skip-pole}, the integral along the small loops around the roots of $\Delta(\kcomplex)$ can be evaluated as the sum of the residues at these roots.
For $\partial\D_\epsilon^+$, the small loops around each root have an angular width $\beta=-\pi$. 
Using \Cref{lemm:skip-pole} and \Cref{re:skip-pole}, the second integral in \eqref{eq:u-opt-effective-D-epsilon} can be evaluated as 
\begin{equation}\label{eq:principle-value-integral-and-residue}
	\begin{aligned}
		\int_{\partial\D_\epsilon^+}f_1(x,t,\kcomplex)\frac{\diff \kcomplex}{2\pi}=\dashint_{-\infty}^{\infty}f_1(x,t,k)\dfrac{\diff k}{2\pi} - \frac{\pi\im}{2\pi}\sum_{m=-\infty}^{\infty} \text{Residue}[f_1(x,t,k)]_{k=k_m},
	\end{aligned}
\end{equation}
where the principal value integral is given by
\begin{equation*}
	\dashint_{-\infty}^{\infty}f_1(x,t,k)\dfrac{\diff k}{2\pi} = \lim_{\substack{R\to\infty\\ M\to\infty\\\epsilon\to0}}
	\left(
		\int_{-R}^{k_{-M}-\epsilon} + \sum_{m=-M}^{M-1}\int_{k_{m}+\epsilon}^{k_{m+1}-\epsilon}
		+\int_{k_M+\epsilon}^{R}
	\right)f_1 \dfrac{\diff k}{2\pi}.
\end{equation*}
The reason we introduce the principal value integral above is that the integrand $f_1(x,t,k)$ has singularities at the real roots of $\Delta(k)$, i.e., $k=k_m$.

We first evaluate the principal value integral and show that it cancels out the first integral in \eqref{eq:u-opt-effective-D-epsilon}.
Explicitly writing $f_1(x,t,k)$, we have
\begin{align}
		f_1(x,t,k) =& \dfrac{2\im e^{-\wcomplex(k)t}\hat{p}(k)}{\Delta(k)}(\sin(kx)e^{\im k L}\statehat_0(k) + \sin(k(L-x)) \statehat_0(-k)) \label{eq:principle-integral-initial}\\ 
		&-\dfrac{4ke^{-\wcomplex(k)t}\hat{p}(k)}{\Delta(k)}\big[\sin(kx)(\tilde{h}_0(\wcomplex(k),t) - \hat{p}(k)\underline{h}_0(\wcomplex(k),t)) +\sin(k(L-x))(\tilde{g}_0(\wcomplex(k),t) - \hat{p}(k)\underline{g}_0(\wcomplex(k),t)) \big]. \label{eq:principle-integral-boundary}
\end{align} 
Let $f_{11}(x,t,k)$ denote the term in \eqref{eq:principle-integral-initial} and $f_{12}(x,t,k)$ denote the term in \eqref{eq:principle-integral-boundary}, i.e., $f_1(x,t,k) = f_{11}(x,t,k) - f_{12}(x,t,k)$.
By $\Delta(k) = -\Delta(-k)$, we have $f_{12}(x,t,-k) = -f_{12}(x,t,k)$ and find that the following integral vanishes,
 $$\dashint_{-\infty}^{\infty}f_{12}(x,t,k)\dfrac{\diff k}{2\pi}=0.$$
 
The integral of $f_{11}(x,t,k)$ can be evaluated as
\begin{equation*}
	\begin{aligned}
		&\dashint_{-\infty}^{\infty}f_{11}(x,t,k)\dfrac{\diff k}{2\pi} \\
		=\ &\dashint_{-\infty}^{\infty}\dfrac{e^{-\wcomplex(k)t}\hat{p}(k)}{\Delta(k)}(e^{\im k (L + x)} - e^{\im k (L - x)})\statehat_0(k)\dfrac{\diff k}{2\pi} +\dashint_{-\infty}^{\infty}\dfrac{e^{-\wcomplex(k)t}\hat{p}(k)}{\Delta(k)}(e^{\im k(L - x)} - e^{-\im k(L - x)})\statehat_0(-k)\dfrac{\diff k}{2\pi}\\
		=\ &\dashint_{-\infty}^{\infty}\dfrac{e^{-\wcomplex(k)t}\hat{p}(k)}{\Delta(k)}(e^{\im k (x + L)} - e^{\im k (x-L)})\statehat_0(k)\dfrac{\diff k}{2\pi}\nonumber \\
		=\ &\dashint_{-\infty}^{\infty}e^{ikx-\wcomplex(k)t}\hat{p}(k)\statehat_0(k)\dfrac{\diff k}{2\pi}\ (\text{from } \Delta(k) = e^{\im k L}\! -\! e^{- \im k L}) \nonumber \\
		=\ &\int_{-\infty}^{\infty}e^{ikx-\wcomplex(k)t}\hat{p}(k)\statehat_0(k)\dfrac{\diff k}{2\pi}, \label{eq:principle-integral-cancel}
\end{aligned}	
\end{equation*}
where the last equality follows from the fact that the integrand in \eqref{eq:principle-integral-cancel} is bounded and well-defined at the roots $k_m,m\in\ZZ$. 
Therefore, the principal value integral of $f_1$ cancels out the first integral in \eqref{eq:u-opt-effective-D-epsilon}.

Now we evaluate the residue of $f_1(x,t,k)$ at $k=k_m$ for \eqref{eq:principle-value-integral-and-residue}.
The function $f_1(x,t,k)$ is the ratio of two other functions, $f_1(x,t,k) = s_1(x,t,k)/\Delta(k)$, see \eqref{eq:principle-integral-initial}--\eqref{eq:principle-integral-boundary}. 
Following \Cref{re:skip-pole}, the residue can be calculated as 
\begin{equation*}
	\text{Residue}[f_1(x,t,k)]_{k=k_m} = \dfrac{s_1(x,t,k_m)}{\Delta'(k_m)},
\end{equation*}
where $\Delta'(k_m)$ is the first derivative of $\Delta(k)$ at $k_m$.
Recall that $\Delta(k) = 2\im\sin(kL)$.
Expanding the residue with $\Delta'(k_m) = \cos(m\pi)2\im L$, we find
\begin{align}
		&\text{Residue}[f_1(x,t,k)]_{k=k_m} \nonumber\\
		&=\dfrac{e^{-\wcomplex(k_m)t}\hat{p}(k_m)}{L\cos(m\pi)} \big[\sin(k_mx)e^{\im k_m L}\statehat_0(k_m) + \sin(k_m(L-x))\statehat_0(-k_m)\big] \label{eq:residue-initial} \\
		&\quad - \dfrac{2k_me^{-\wcomplex(k_m)t}\hat{p}(k_m)}{\im L\cos(m\pi)} \nonumber \\ & \quad\big[\sin(k_mx)(\tilde{h}_0(\wcomplex(k_m),t)  - \hat{p}(k_m)\underline{h}_0(\wcomplex(k_m),t)) + \sin(k_m(L-x))(\tilde{g}_0(\wcomplex(k_m),t) - \hat{p}(k_m)\underline{g}_0(\wcomplex(k_m),t))\big]. \label{eq:residue-boundary}
\end{align}
Let $R_1$ denote the term in \eqref{eq:residue-initial} and $R_2$ denote the term in \eqref{eq:residue-boundary}, i.e., $\text{Residue}[f_1(x,t,k)]_{k=k_m} = R_1(x,t,k_m) - R_2(x,t,k_m)$.
Replacing $\sin(kx)$ and $\sin(k_m(L-x))$ with equivalent exponential functions in $R_1$, we have
\begin{equation*}
	R_1(x,t,k_m) = \dfrac{1}{\im L} e^{\im k_m x-\wcomplex(k_m)t}\hat{p}(k_m)\big[\statehat_0(k_m) - \statehat_0(-k_m)\big].
\end{equation*}
By the definition of the unified transform \eqref{eq:unified-transform}, we find
\begin{equation*}
	\statehat_0(k_m) - \statehat_0(-k_m) = -2\im\int_{0}^{L}\sin(k_m\xi)\state_0(\xi)\diff \xi.
\end{equation*}
Therefore, $R_1(x,t,k_m)$ can be rewritten as
$$	R_1(x,t,k_m) = -e^{\im k_m x-\wcomplex(k_m)t}\hat{p}(k_m) \state_m^\circ, $$
where
$$
	\state_m^\circ = \frac{2}{L}\int_{0}^{L}\sin(k_m\xi)\state_0(\xi)\diff \xi.
$$
Recall $\check{h}(k_m,t) = \tilde{h}_0(\wcomplex(k_m),t)  - \hat{p}(k_m)\underline{h}_0(\wcomplex(k_m),t)$ and $\check{g}(k_m,t) = \tilde{g}_0(\wcomplex(k_m),t) - \hat{p}(k_m)\underline{g}_0(\wcomplex(k_m),t)$.
Replacing $\sin(k_mx)$, $\sin(k_m(L-x))$, and $\cos(m\pi)$ with equivalent exponential functions in $R_2$, we obtain
	$$R_2(x,t,k_m)  = e^{\im k_m x-\wcomplex(k_m)t}\hat{p}(k_m)b_m(t),$$
	where
	$$
	b_m(t)  = \dfrac{2}{L}k_m(\check{g}(k_m,t)-\cos(m\pi)\check{h}(k_m,t)).
	$$
Repeating the analysis for the integral of $f_2$ in \eqref{eq:u-opt-effective-D-epsilon} and combining the results, we find
\begin{equation}\label{eq:u-opt-series}
	\begin{aligned}
		u^*(x,t) &=\frac{\im}{2} \sum_{m=-\infty}^{\infty}e^{\im k_m x-\wcomplex(k_m)t}\hat{p}(k_m)(-\state_m^\circ - b_m(t))- \frac{\im}{2}\sum_{m=-\infty}^{\infty}e^{\im k_m x}\ \underline{b}_m(t),
	\end{aligned}
\end{equation}
where $\underline{b}_m(t)$ is defined in \eqref{eq:underline-bn}.

Repeating the same steps for \eqref{eq:state-expression}, we obtain
\begin{equation}\label{eq:state-expression-series}
	\begin{aligned}
		\state(x,t) &=\frac{\im}{2} \sum_{m=-\infty}^{\infty}e^{\im k_m x-\wcomplex(k_m)t}(-\state_m^\circ - b_m(t)) = \sum_{m=-\infty}^{\infty}e^{\im k_m x}\psi_m(t),
	\end{aligned}
\end{equation}
where $\psi_m(t) = \frac{\im}{2}e^{-\wcomplex(k_m)t}(-\state_m^\circ - b_m(t))$.
This is a Fourier series expansion of the state with coefficients $\psi_m$.
Therefore, optimal control \eqref{eq:u-opt-series} can be rewritten as 
\begin{equation*}
		u^*(x,t) = \sum_{m=-\infty}^{\infty}e^{\im k_m x}\hat{p}(k_m)\psi_m(t) - \frac{\im}{2}\sum_{m=-\infty}^{\infty}e^{\im k_m x}\ \underline{b}_m(t).
\end{equation*}
By the convolution theorem for the Fourier series, we obtain the following convolution form,
\begin{equation}\label{eq:u-opt-convolution-minus-L-to-L}
	u^*(x,t)= \int_{-L}^{L} \Gamma(x,\xi)\state(\xi,t)\frac{\diff \xi}{2L}- \frac{\im}{2}\sum_{m=-\infty}^{\infty}e^{\im k_m x}\ \underline{b}_m(t),
\end{equation}
where $\Gamma(x,\xi)$ is defined in \eqref{eq:Gamma-convolution}.
The expression \eqref{eq:u-opt-convolution-minus-L-to-L} depends on the state $\state(\xi,t)$ at both $\xi\in(-L,0)$ and $\xi\in(0,L)$.
Lastly, we use the symmetry of the state to rewrite the convolution integral in \eqref{eq:u-opt-convolution-minus-L-to-L} as an integral in $(0,L)$.

Using $\state_m^\circ = -\state_{-m}^\circ, b_m = -b_{-m}$ and $\state_0^\circ = b_0=0$, from \eqref{eq:state-expression-series} we have
\begin{equation*}
	\begin{aligned}
		\state(x,t) = \sum_{m=1}^{\infty}e^{-\wcomplex(k_m)t}\sin(k_mx)(\state_m^\circ + b_m(t)).
	\end{aligned}
\end{equation*}
This is a classical sine series solution to boundary value problems of the reaction-diffusion equation, and thus $\state(-x,t) = -\state(x,t)$.
Using $\underline{b}_m(t) = - \underline{b}_{-m}(t)$ and $\underline{b}_0(t)=0$, we find
\begin{equation*}
	\frac{\im}{2}\sum_{m=-\infty}^{\infty}e^{\im k_m x}\ \underline{b}_m(t) = -\sum_{m=1}^{\infty}\sin(k_mx) \underline{b}_m(t).
\end{equation*}
The final expression \eqref{eq:u-opt-convolution} then follows from separating the convolution integral in \eqref{eq:u-opt-convolution-minus-L-to-L} into $(0,L)$ and $(-L,0)$ and using the fact that $\state(-x,t) = -\state(x,t)$.
\end{proof}

\section{Numerical experiments}
\label{sec:numerical-experiment}

\begin{figure}[htb]
	\centering
\begin{subfigure}[b]{0.49\textwidth}
	\centering
	\includegraphics[width=0.8\linewidth]{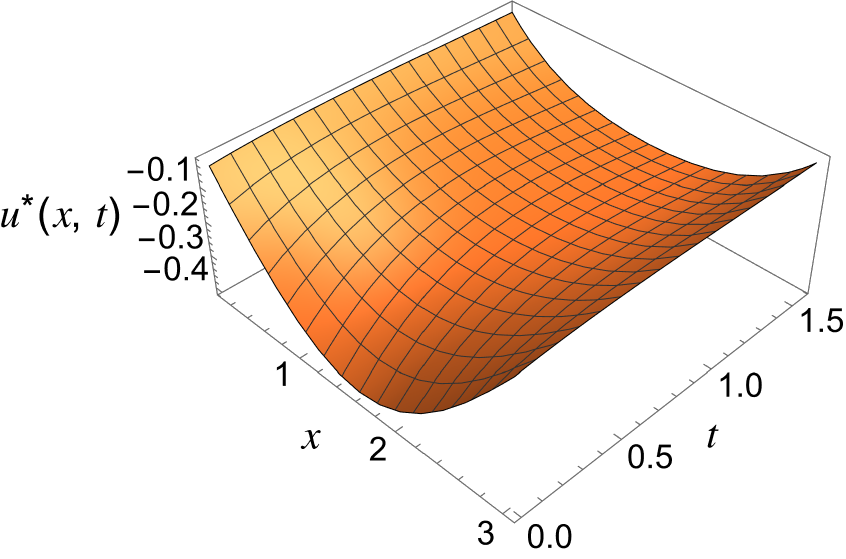}
	\caption{Numerical computation of $\state^*(x,t)$ using \eqref{eq:state-expression}}
	\label{fig:numerical-state}
\end{subfigure}\hfill
\begin{subfigure}[b]{0.49\textwidth}
	\centering
	\includegraphics[width=0.8\linewidth]{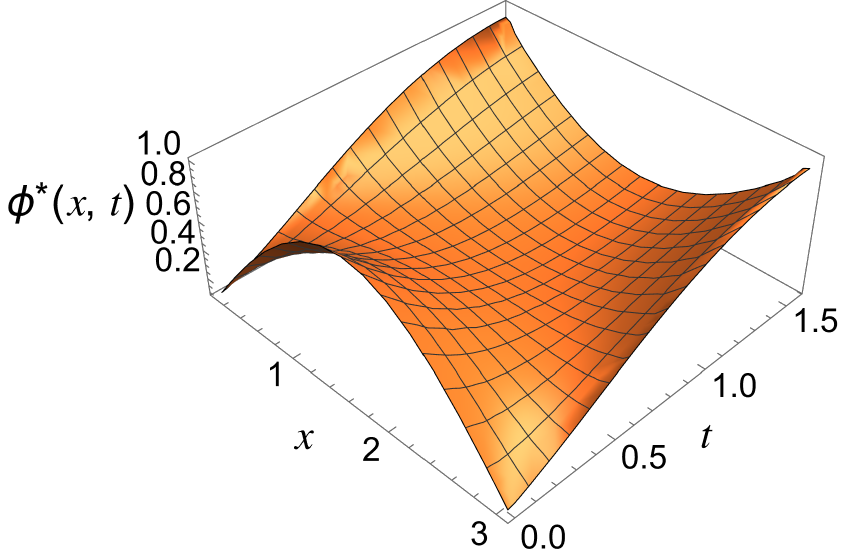}
	\caption{Numerical computation of $u^*(x,t)$ using \eqref{eq:u-opt-effective}}
	\label{fig:numerical-control}
\end{subfigure}
	\caption{Numerical computation of the state and control}
	\label{fig:numerical-experiment}
\end{figure}
\begin{figure}[th]
	\centering
	\begin{subfigure}{0.33\textwidth}
		\centering
		\includegraphics[width=0.8\linewidth]{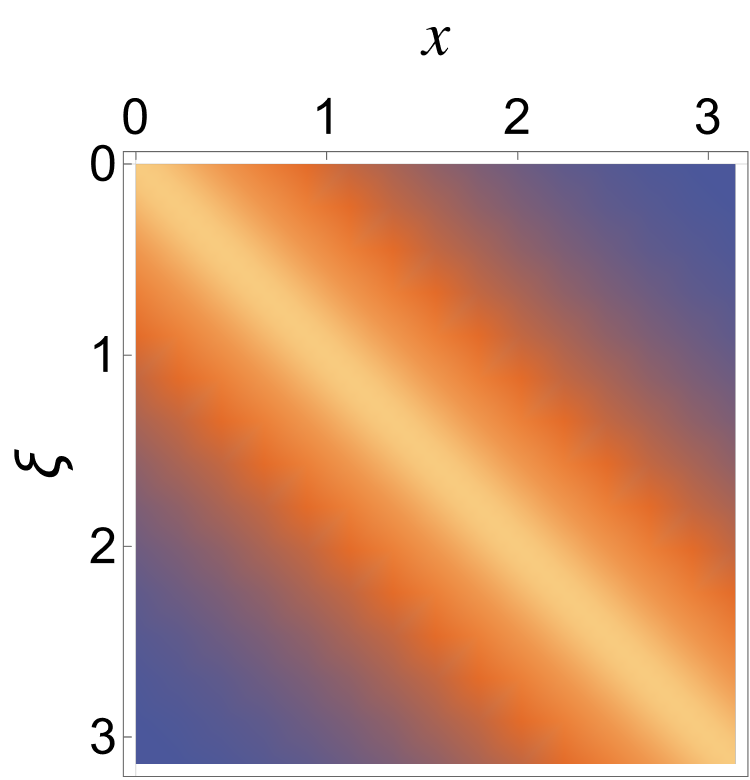}
		\caption{$\Gamma(x,\xi)$ with $c=0$}
		\label{fig:Toeplitz-c-0}
	\end{subfigure}\hfill
	\begin{subfigure}{0.33\textwidth}
		\centering
		\includegraphics[width=0.8\linewidth]{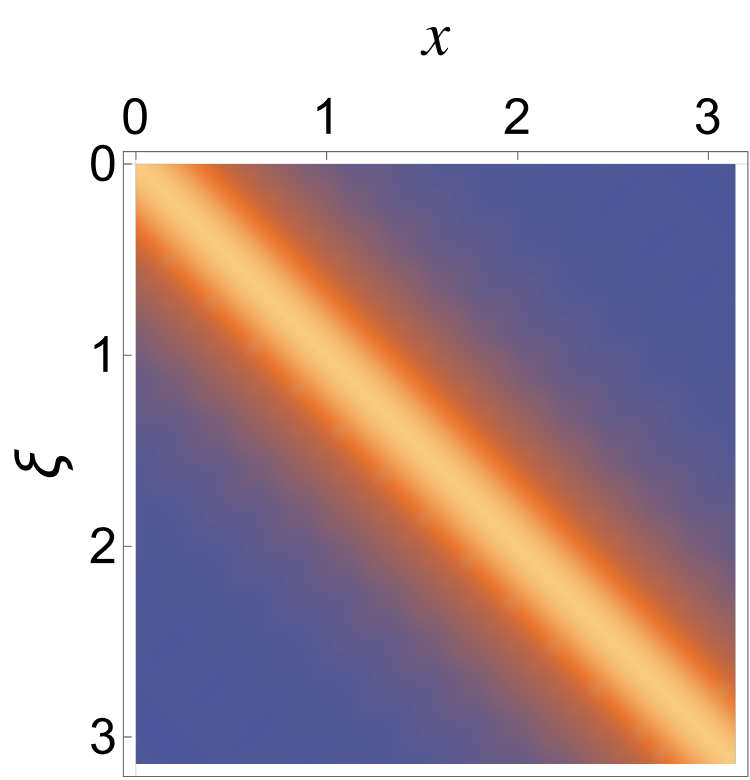}
		\caption{ $\Gamma(x,\xi)$ with $c=5$}
		\label{fig:Toeplitz-c-5}
	\end{subfigure}\hfill
	\begin{subfigure}{0.33\textwidth}
		\centering
		\includegraphics[width=0.8\linewidth]{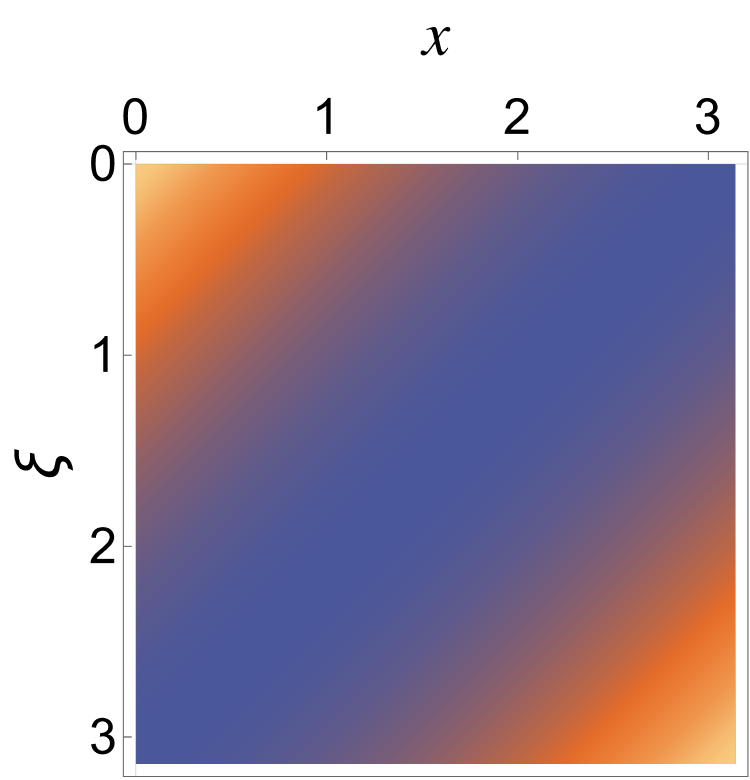}
		\caption{ $\Gamma(x,-\xi)$ with $c=0$}
		\label{fig:Hankel-c-0}
	\end{subfigure}\hfill
	\begin{subfigure}{0.33\textwidth}
		\centering
		\includegraphics[width=0.8\linewidth]{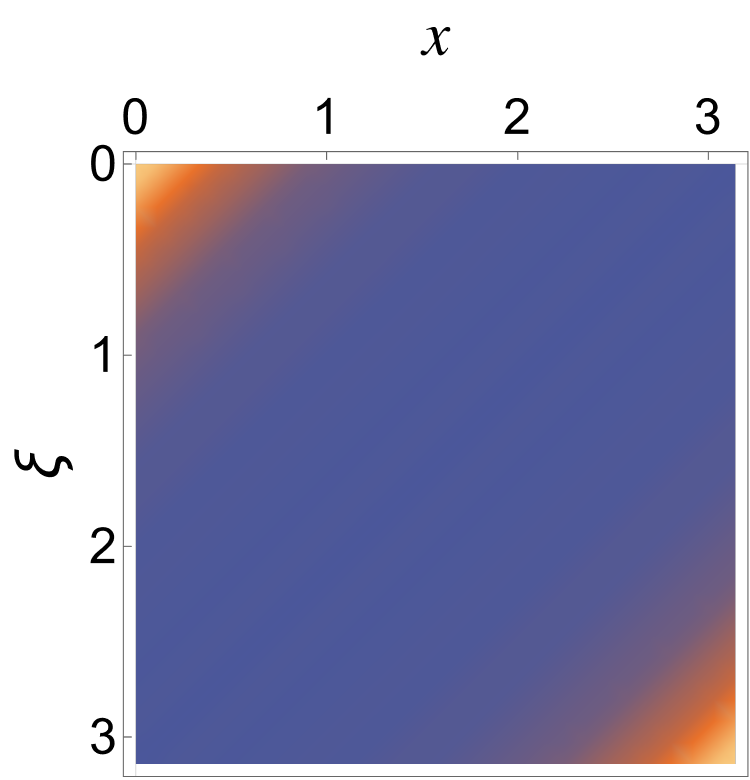}
		\caption{ $\Gamma(x,-\xi)$ with $c=5$}
		\label{fig:Hankel-c-5}
	\end{subfigure}
	\begin{subfigure}{0.33\textwidth}
		\centering
		\includegraphics[width=0.8\linewidth]{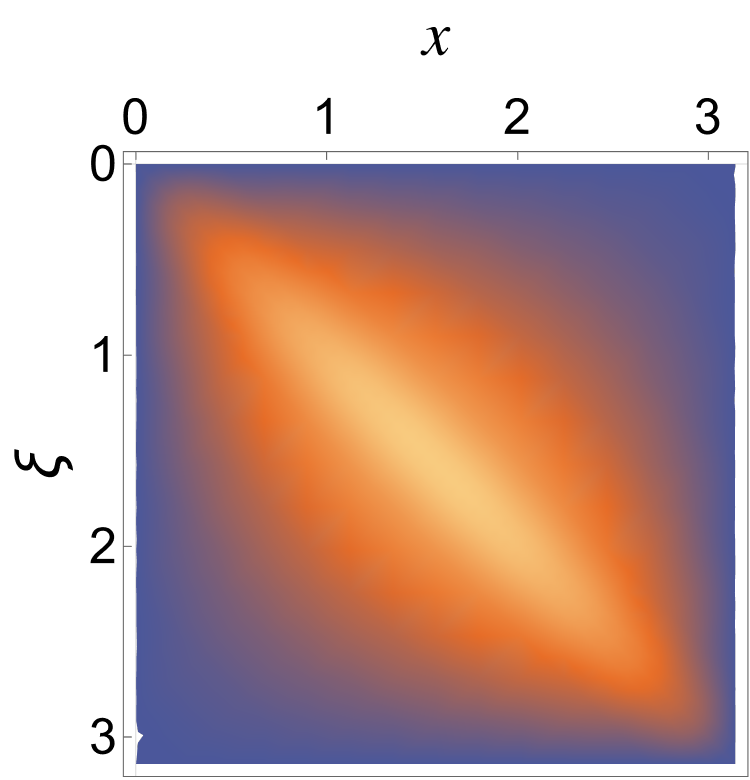}
		\caption{$\Gamma(x,\xi) -\Gamma(x,-\xi)$ with $c=0$}
		\label{fig:Kernel-c-0}
	\end{subfigure}\hfill
	\begin{subfigure}{0.33\textwidth}
		\centering
		\includegraphics[width=0.8\linewidth]{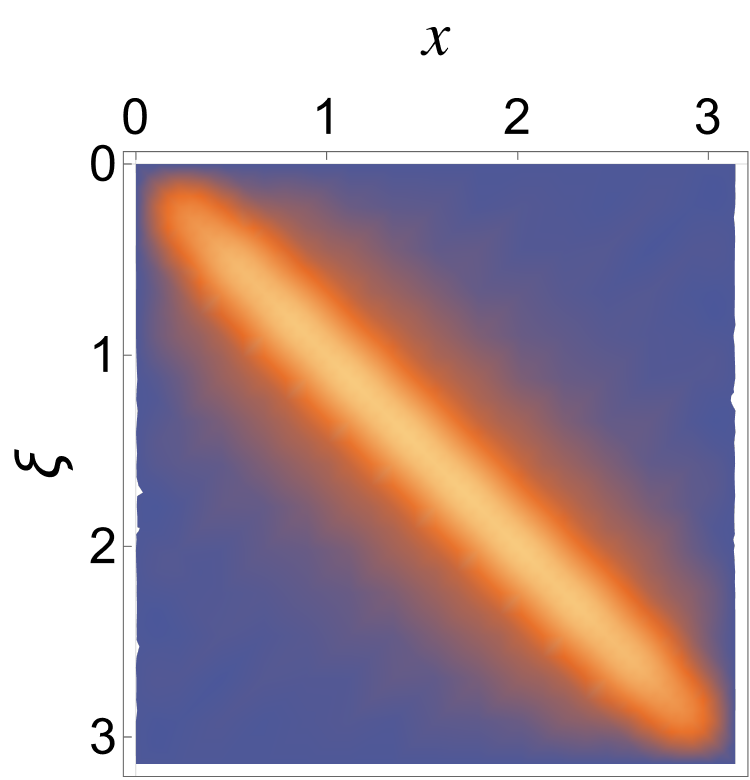}
		\caption{$\Gamma(x,\xi) - \Gamma(x,-\xi)$ with $c=5$}
		\label{fig:Kernel-c-5}
	\end{subfigure}
	\caption{Structures of the feedback kernel illustrated by color. A brighter color represents a higher function value.}
	\label{fig:structure-kernel}
\end{figure}
We have shown that the integral representation \eqref{eq:u-opt-effective} is equivalent to a series representation \eqref{eq:u-opt-series} that can be rewritten in feedback form \eqref{eq:u-opt-convolution}. Although the series representation can reveal the structure of the control, the integral representation is favorable for numerical computations. 
We use \texttt{Mathematica} to directly compute the state and control expressions to compare the numerical properties of the two representations.
% \begin{figure}[t]
% 	\centering
% 	\includegraphics[width=\columnwidth]{"figure/Control-comparison"}
% 	\caption{Performance of the optimal control}
% 	\label{fig:control-comparison}
% \end{figure}
We consider the heat equation with initial condition $\state_0(x)=\sin(x)$ and boundary condition $\state(0,t) = \state(L,t) = \sin(t)$.
\Cref{fig:numerical-experiment} shows the numerical computation using the integral representations of the state \eqref{eq:state-expression} and the control \eqref{eq:u-opt-effective}.
The infinite series \eqref{eq:state-expression-series} and \eqref{eq:u-opt-series} are not numerically convergent and cannot be computed without approximations.
In contrast, the integral representations  \eqref{eq:u-opt-effective} and \eqref{eq:state-expression} converge in existing numerical integration solvers and can be computed very efficiently. 
Numerical advantages of complex integral representations have already been discussed for general boundary conditions and other linear PDEs, see for example \cite{de2019hybrid}.

% Next, we consider a nonzero boundary conditions $g_0(t)=10$ at $x=0$ and $h_0(t)=10$ at $x=L$ with $L=\pi$. The initial condition is $\state_0(x)=10$.  We compare the optimal control to a baseline control \eqref{eq:u-opt-convolution} with $\underline{b}_m(t)=0$. 
% This baseline control was first derived in \cite{epperlein2016spatially} for the case of homogeneous boundary conditions. 
% \Cref{fig:control-comparison} shows that our optimal control can improve the state performance at the middle point $x=L/2$ by 40\% compared to the baseline control.

We also investigate the structure of the kernel function in \eqref{eq:u-opt-convolution}. We approximate the infinite sum with $m$ from $-10$ to $10$ and compute the functions $\Gamma(x,\xi)$ and $\Gamma(x,-\xi)$ with $c=0$ and $c=5$. Note that the approximation is accurate enough to reveal the structure since $\hat{p}(m)$ is small when $\abs{m}>10$. \Cref{fig:structure-kernel} shows the diagonal and anti-diagonal structures of the kernel functions $\Gamma(x,\xi)$ and $\Gamma(x,-\xi)$, respectively.
The color represents the function value; a brighter color means a higher value.
 Following \Cref{re:structure}, $\Gamma(x,\xi)$ has a Toeplitz structure and $\Gamma(x,-\xi)$ has a Hankel structure. \cite{bamieh2002distributed} showed that only the Toeplitz structure is present in the kernel when the spatial domain is unbounded and the kernel decays exponentially. In the case of bounded spatial domains, \Cref{fig:Kernel-c-0} shows that the kernel for the heat equation with $c=0$ does not decay exponentially due to the presence of the Hankel structure in \Cref{fig:Hankel-c-0}.
 However, the effect of Hankel structure also depends on the value of $c$. As $c$ increases, the Toeplitz function $\Gamma(x,\xi)$ is more concentrated, and thus the decay rate is higher, see \Cref{fig:Toeplitz-c-5}; the Hankel function $\Gamma(x,-\xi)$ is more spread out, and the anti-diagonal values are more concentrated at the boundaries, see \Cref{fig:Hankel-c-5}. 
 By comparing \Cref{fig:Kernel-c-0,fig:Kernel-c-5}, the two changes suggest that optimal control has a more decentralized structure as $c$ increases.

\section{Conclusions}\label{sec:conclusion}
We derived integral representations for LQR control of linear evolution PDEs with general smooth boundary conditions. 
The integral representation on the real line depends on unknown boundary values.
We presented a general procedure for deforming integrals from the real line to well-constructed contours in the complex plane along which the contribution of the unknown boundary values vanishes. 
The procedure for contour deformation can be applied to PDEs with $w_{\text{Im}}(\kcomplex)=0$. 

For illustration, we derived complex integral representations for optimal control of the reaction-diffusion equation with Dirichlet boundary conditions.
The resulting integral representation was then rewritten as a Fourier series, leading to a convolution feedback form of optimal control.
We found that the kernel function in the convolution feedback form preserves the Toeplitz plus Hankel structure discovered in \cite{epperlein2016spatially} in the case of homogeneous boundary conditions.
In numerical experiments, we illustrated that the integral representation is numerically convergent, whereas the series representation cannot be computed without approximation.
We also found that the kernel function has a varying decay rate depending on the reaction coefficients.

Future work includes finding contour deformation for PDEs with $w_{\text{Im}}(\kcomplex)\neq0$. Also, deriving feedback forms and analyzing the corresponding structural properties for general linear evolution PDEs beyond the reaction-diffusion equation will be pursued.
Another extension is to derive optimal control of systems beyond evolution PDEs for which the Fokas method has been successful, such as the wave equation and elliptic equations \cite{fokasbook}.
In addition to optimal control, optimal estimation can also be analyzed using frequency domain methods \cite{arbelaiz2024optimal}, and thus our complex spatial frequency approach can be applied to optimal estimation in a similar way.
For applications, it is of interest to implement our approach to variable speed limit control on freeways \cite{delle2017traffic,block2024lq}, advection-diffusion-reaction systems \cite{glowinski2022bilinear}, and soft robotics \cite{della2023model}.

\section*{Acknowledgments}
Z. Li thanks Prof. M. Jovanovic for helpful discussions about \Cref{sec:state-transformation,sec:separation-variables}.

\bibliographystyle{IEEEtran}
\bibliography{main-finite-interval-journal-one-column.bib}

% Generated by IEEEtran.bst, version: 1.14 (2015/08/26)
\begin{thebibliography}{10}
\providecommand{\url}[1]{#1}
\csname url@samestyle\endcsname
\providecommand{\newblock}{\relax}
\providecommand{\bibinfo}[2]{#2}
\providecommand{\BIBentrySTDinterwordspacing}{\spaceskip=0pt\relax}
\providecommand{\BIBentryALTinterwordstretchfactor}{4}
\providecommand{\BIBentryALTinterwordspacing}{\spaceskip=\fontdimen2\font plus
\BIBentryALTinterwordstretchfactor\fontdimen3\font minus \fontdimen4\font\relax}
\providecommand{\BIBforeignlanguage}[2]{{%
\expandafter\ifx\csname l@#1\endcsname\relax
\typeout{** WARNING: IEEEtran.bst: No hyphenation pattern has been}%
\typeout{** loaded for the language `#1'. Using the pattern for}%
\typeout{** the default language instead.}%
\else
\language=\csname l@#1\endcsname
\fi
#2}}
\providecommand{\BIBdecl}{\relax}
\BIBdecl

\bibitem{liu2024adjoint}
X.~Liu and J.~F. MacArt, ``Adjoint-based machine learning for active flow control,'' \emph{Physical Review Fluids}, vol.~9, no.~1, p. 013901, 2024.

\bibitem{lenhart1993optimal}
S.~Lenhart and D.~Wilson, ``Optimal control of a heat transfer problem with convective boundary condition,'' \emph{Journal of optimization theory and applications}, vol.~79, pp. 581--597, 1993.

\bibitem{shapiro2017model}
C.~R. Shapiro, P.~Bauweraerts, J.~Meyers, C.~Meneveau, and D.~F. Gayme, ``Model-based receding horizon control of wind farms for secondary frequency regulation,'' \emph{Wind Energy}, vol.~20, no.~7, pp. 1261--1275, 2017.

\bibitem{fister2003optimal}
K.~R. Fister and C.~M. McCarthy, ``Optimal control of a chemotaxis system,'' \emph{Quarterly of Applied Mathematics}, vol.~61, no.~2, pp. 193--211, 2003.

\bibitem{smyshlyaev2010adaptive}
A.~Smyshlyaev and M.~Krstic, \emph{Adaptive control of parabolic PDEs}.\hskip 1em plus 0.5em minus 0.4em\relax Princeton University Press, 2010.

\bibitem{ascencio2017backstepping}
P.~Ascencio, A.~Astolfi, and T.~Parisini, ``Backstepping pde design: A convex optimization approach,'' \emph{IEEE Transactions on automatic control}, vol.~63, no.~7, pp. 1943--1958, 2017.

\bibitem{lions1971optimal}
J.~L. Lions, \emph{Optimal control of systems governed by partial differential equations}.\hskip 1em plus 0.5em minus 0.4em\relax Springer, 1971, vol. 170.

\bibitem{hinze2008optimization}
M.~Hinze, R.~Pinnau, M.~Ulbrich, and S.~Ulbrich, \emph{Optimization with PDE constraints}.\hskip 1em plus 0.5em minus 0.4em\relax Springer Science \& Business Media, 2008, vol.~23.

\bibitem{anderson2007optimal}
B.~D. Anderson and J.~B. Moore, \emph{Optimal control: linear quadratic methods}.\hskip 1em plus 0.5em minus 0.4em\relax Courier Corporation, 2007.

\bibitem{curtain2012introduction}
R.~F. Curtain and H.~Zwart, \emph{An introduction to infinite-dimensional linear systems theory}.\hskip 1em plus 0.5em minus 0.4em\relax Springer Science \& Business Media, 2012, vol.~21.

\bibitem{bamieh2002distributed}
B.~Bamieh, F.~Paganini, and M.~A. Dahleh, ``Distributed control of spatially invariant systems,'' \emph{IEEE Transactions on automatic control}, vol.~47, no.~7, pp. 1091--1107, 2002.

\bibitem{d2003distributed}
R.~D'Andrea and G.~E. Dullerud, ``Distributed control design for spatially interconnected systems,'' \emph{IEEE Transactions on automatic control}, vol.~48, no.~9, pp. 1478--1495, 2003.

\bibitem{motee2008optimal}
N.~Motee and A.~Jadbabaie, ``Optimal control of spatially distributed systems,'' \emph{IEEE Transactions on Automatic Control}, vol.~53, no.~7, pp. 1616--1629, 2008.

\bibitem{jovanovic2005ill}
M.~R. Jovanovic and B.~Bamieh, ``On the ill-posedness of certain vehicular platoon control problems,'' \emph{IEEE Transactions on Automatic Control}, vol.~50, no.~9, pp. 1307--1321, 2005.

\bibitem{langbort2005distributed}
C.~Langbort and R.~D'Andrea, ``Distributed control of spatially reversible interconnected systems with boundary conditions,'' \emph{SIAM journal on control and optimization}, vol.~44, no.~1, pp. 1--28, 2005.

\bibitem{epperlein2016spatially}
J.~P. Epperlein and B.~Bamieh, ``Spatially invariant embeddings of systems with boundaries,'' in \emph{2016 American Control Conference (ACC)}.\hskip 1em plus 0.5em minus 0.4em\relax IEEE, 2016, pp. 6133--6139.

\bibitem{fokas1997unified}
A.~S. Fokas, ``A unified transform method for solving linear and certain nonlinear pdes,'' \emph{Proceedings of the Royal Society of London. Series A: Mathematical, Physical and Engineering Sciences}, vol. 453, no. 1962, pp. 1411--1443, 1997.

\bibitem{fokas2008unified}
------, \emph{A unified approach to boundary value problems}.\hskip 1em plus 0.5em minus 0.4em\relax SIAM, 2008.

\bibitem{fokasbook}
\BIBentryALTinterwordspacing
A.~Fokas and E.~Kaxiras, \emph{Modern Mathematical Methods for Scientists and Engineers}.\hskip 1em plus 0.5em minus 0.4em\relax WORLD SCIENTIFIC (EUROPE), 2023. [Online]. Available: \url{https://www.worldscientific.com/doi/abs/10.1142/q0348}
\BIBentrySTDinterwordspacing

\bibitem{deconinck2014method}
B.~Deconinck, T.~Trogdon, and V.~Vasan, ``The method of fokas for solving linear partial differential equations,'' \emph{SIAM Review}, vol.~56, no.~1, pp. 159--186, 2014.

\bibitem{kalimeris2023numerical}
K.~Kalimeris, T.~{\"O}zsarl, and N.~Dikaios, ``Numerical computation of neumann controls for the heat equation on a finite interval,'' \emph{IEEE Transactions on Automatic Control}, 2023.

\bibitem{Mathematica}
\BIBentryALTinterwordspacing
W.~R. Inc., ``Mathematica, {V}ersion 14.2,'' champaign, IL, 2024. [Online]. Available: \url{https://www.wolfram.com/mathematica}
\BIBentrySTDinterwordspacing

\bibitem{arbelaiz2024optimal}
J.~Arbelaiz, B.~Bamieh, A.~E. Hosoi, and A.~Jadbabaie, ``Optimal estimation in spatially distributed systems: How far to share measurements from?'' \emph{IEEE Transactions on Automatic Control}, vol.~70, no.~5, pp. 3226--3239, 2025.

\bibitem{Ayamou2024finite}
M.~Ayamou, N.~Espitia, A.~Polyakov, and E.~Fridman, ``Finite-dimensional homogeneous boundary control for a 1d reaction-diffusion equation,'' in \emph{2024 IEEE 63rd Conference on Decision and Control (CDC)}, 2024, pp. 1962--1967.

\bibitem{papanicolaou2023orthonormal}
N.~C. Papanicolaou and I.~Christov, ``Orthonormal eigenfunction expansions for sixth-order boundary value problems,'' in \emph{Journal of Physics: Conference Series}, vol. 2675, no.~1.\hskip 1em plus 0.5em minus 0.4em\relax IOP Publishing, 2023, p. 012016.

\bibitem{liberzon2011calculus}
D.~Liberzon, \emph{Calculus of variations and optimal control theory: a concise introduction}.\hskip 1em plus 0.5em minus 0.4em\relax Princeton university press, 2011.

\bibitem{bernstein1993some}
D.~S. Bernstein and W.~So, ``Some explicit formulas for the matrix exponential,'' \emph{IEEE Transactions on Automatic Control}, vol.~38, no.~8, pp. 1228--1232, 1993.

\bibitem{ablowitz2003complex}
M.~J. Ablowitz and A.~S. Fokas, \emph{Complex variables: introduction and applications}.\hskip 1em plus 0.5em minus 0.4em\relax Cambridge University Press, 2003.

\bibitem{de2019hybrid}
F.~De~Barros, M.~Colbrook, and A.~Fokas, ``A hybrid analytical-numerical method for solving advection-dispersion problems on a half-line,'' \emph{International Journal of Heat and Mass Transfer}, vol. 139, pp. 482--491, 2019.

\bibitem{delle2017traffic}
M.~L. Delle~Monache, B.~Piccoli, and F.~Rossi, ``Traffic regulation via controlled speed limit,'' \emph{SIAM Journal on Control and Optimization}, vol.~55, no.~5, pp. 2936--2958, 2017.

\bibitem{block2024lq}
B.~Block and S.~Stockar, ``Lq control of traffic flow models via variable speed limits,'' in \emph{2024 American Control Conference (ACC)}.\hskip 1em plus 0.5em minus 0.4em\relax IEEE, 2024, pp. 4262--4267.

\bibitem{glowinski2022bilinear}
R.~Glowinski, Y.~Song, X.~Yuan, and H.~Yue, ``Bilinear optimal control of an advection-reaction-diffusion system,'' \emph{SIAM Review}, vol.~64, no.~2, pp. 392--421, 2022.

\bibitem{della2023model}
C.~Della~Santina, C.~Duriez, and D.~Rus, ``Model-based control of soft robots: A survey of the state of the art and open challenges,'' \emph{IEEE Control Systems Magazine}, vol.~43, no.~3, pp. 30--65, 2023.

\end{thebibliography}

\section*{Appendix}
Here, we collect some results from complex analysis used in the proofs.

\begin{lemma}[{Cauchy's Theorem \cite[Theorem 2.5.2]{ablowitz2003complex}}]\label{lemm:cauchy}
	
	If a function $f$ is analytic in a simply connected domain $\D$, then along a simple closed contour $\C$ in $\D$
	\begin{equation*}
		\oint_{\C}f(z)\diff z=0.
	\end{equation*}
\end{lemma}
\begin{lemma}[{Jordan's Lemma \cite[Lemma 4.2.2]{ablowitz2003complex}}]\label{lemm:jordan}
	
	Let $\C_R$ be a circular arc of radius $R$ centered at the origin and lying on the upper-half complex plane defined by $z=Re^{\im\theta}$ with $0\leq\theta_1\leq\theta\leq\theta_2\leq\pi$.	
	Suppose that on the circular arc $\C_R$, we have $f(z)\to0$ uniformly as $R\to\infty$. Then
	\begin{equation*}
		\lim_{R\to\infty}\int_{\C_R} e^{\im\lambda z}f(z)\diff z=0,\quad \lambda>0.
	\end{equation*}
\end{lemma}

% \begin{comment}
\begin{remark}
	If $\lambda$ is negative, then Jordan's lemma is still valid, provided that $\C_R$ is defined on the lower-half complex plane.
\end{remark}

%An alternative of Jordan's lemma is the following, which is useful for circular arcs that do not intersect the real line.
%\begin{lemma}[\!{\cite[Theorem 4.3.2]{ablowitz2003complex}}]
%	\label{lemm:jordan-complex}
%	If on a circular arc $\C_R$ of radius $R$ centered at the origin, $zf(z)\to0$ uniformly as $R\to\infty$, then
%	\begin{equation*}
%		\lim_{R\to\infty}\int_{\C_R} f(z)\diff z=0.
%	\end{equation*}
%\end{lemma}
% \end{comment}
\begin{lemma}[\!{\cite[Theorem 4.3.1]{ablowitz2003complex}}]
	\label{lemm:skip-pole}
		Consider a function $f(z)$ analytic in an annulus centered at $z_0$. Let $\C_{\epsilon,\beta[z_0]}$ be a circular arc of radius $\epsilon$ and angular width $\beta$, which lies entirely within the annulus centered at $z_0: z=z_0+\epsilon e^{\im \alpha}$, with $\alpha_0\leq\alpha\leq\alpha_0+\beta$. 
%			\item[a.] 
%			Suppose that on the contour  $\C_{\epsilon,\beta[z_0]}$, we have $(z-z_0)f(z)\to0$ uniformly as $\epsilon\to0$. Then
%			\begin{equation*}
%				\lim_{\epsilon\to0}\int_{\C_{\epsilon,\beta[z_0]}} f(z)\diff z = 0.
%			\end{equation*}
%			\item[b.] 
		Let $f(z)=\sum_{n=-\infty}^{\infty}C_n(z-z_0)^n$ denote the Laurent series expansion around the point $z_0$.
		If $f(z)$ has a simple pole at $z_0$, i.e., $C_{-1}\neq0$ and $C_{n}=0$ for all $n\leq-2$, then
			\begin{equation*}
				\lim_{\epsilon\to0}\int_{\C_{\epsilon,\beta[z_0]}} f(z)\diff z=\im\beta\ \text{\normalfont Residue}\left[f(z)\right]_{z=z_0},
			\end{equation*}
			where $\text{\normalfont Residue}\left[f(z)\right]_{z=z_0} =C_{-1}$.
\end{lemma}
\begin{remark}\label{re:skip-pole}
	This remark follows from \cite[eq. 4.1.10]{ablowitz2003complex}.
	Suppose that the function $f(z)$ is given by a ratio of two functions $N(z)$ and $D(z)$,  i.e., $f(z) = N(z)/D(z)$, where both functions are analytic in the neighborhood of $z_0$.
	If $N(z_0)\neq0, D(z_0)=0,$ and $D'(z_0)\neq0$, then  $f(z)$ has a simple pole at $z_0$ and
	$$ \text{\normalfont Residue}\left[f(z)\right]_{z=z_0} =\frac{N(z_0)}{D'(z_0)} .$$

\end{remark}
 
\begin{IEEEbiography}[{\includegraphics[width=1.1in,height=1.3in,clip,keepaspectratio]{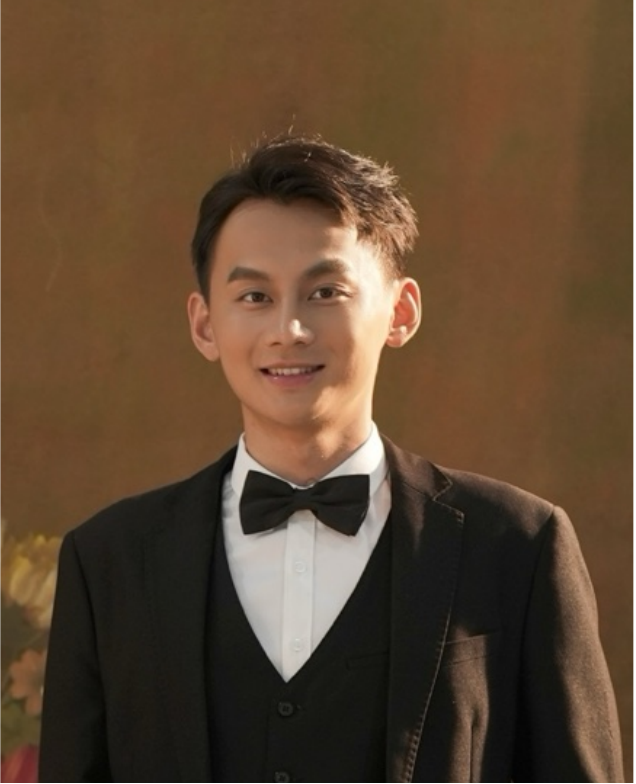}}]{Zhexian Li} (Student Member, IEEE) received a B.Eng. degree in Transportation Engineering from Southeast University, Nanjing, China, in 2020.
He is currently a Ph.D. candidate with the Department of Civil and Environmental Engineering at the University of Southern California. His research interests include optimal control of PDEs and model predictive control, with applications to civil infrastructure systems, such as roadway transportation systems and solute transport in fluids.
	\end{IEEEbiography}
	\begin{IEEEbiography}[{\includegraphics[width=1.1in,height=1.3in,clip,keepaspectratio]{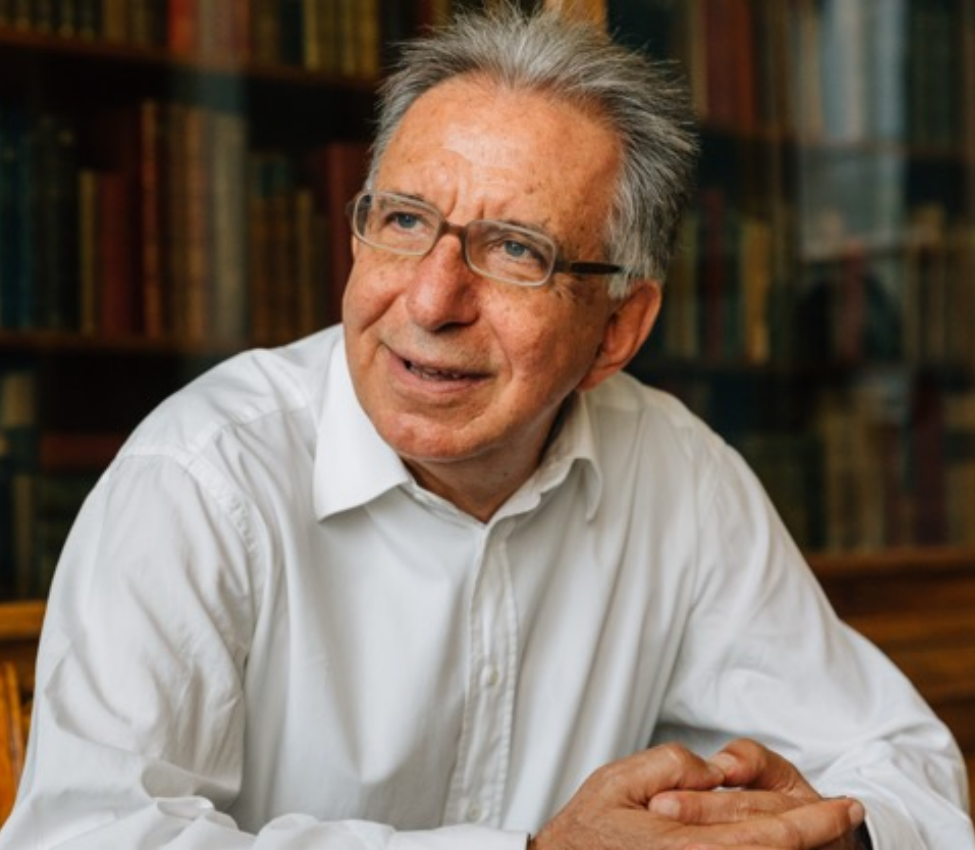}}]{Athanassios S. Fokas} received a B.S. in Aeronautics from Imperial College, a Ph.D. in Applied Mathematics from California Institute of Technology, and an M.D. from the University of Miami. He also has eight honorary degrees. 
	
	In 1995 he was appointed to a Chair in Applied Mathematics at Imperial College and in 2002 he became the first holder of the inaugural Chair of Nonlinear Mathematical Science at the Department of Applied Mathematics and Theoretical Physics of the University of Cambridge.
	
	He is a member of the Academy of Athens and of all three major European Academies, including Academia Europaea. He is a Fellow of Clare Hall College, Cambridge, of the Guggenheim Foundation, of the American Institute for Medical and Biological Engineering, and of the American Mathematical Society.
	Among his many awards are  the Naylor Prize of the London Mathematical Society, and the Blaise Pascal Medal of the European Academy of Sciences, ``\textit{for the Fokas method, which is considered the most important development in the solution of partial differential equations since the works of Fourier, Laplace and Cauchy}''. 
	\end{IEEEbiography}
	\begin{IEEEbiography}[{\includegraphics[width=1.1in,height=1.3in,clip,keepaspectratio]{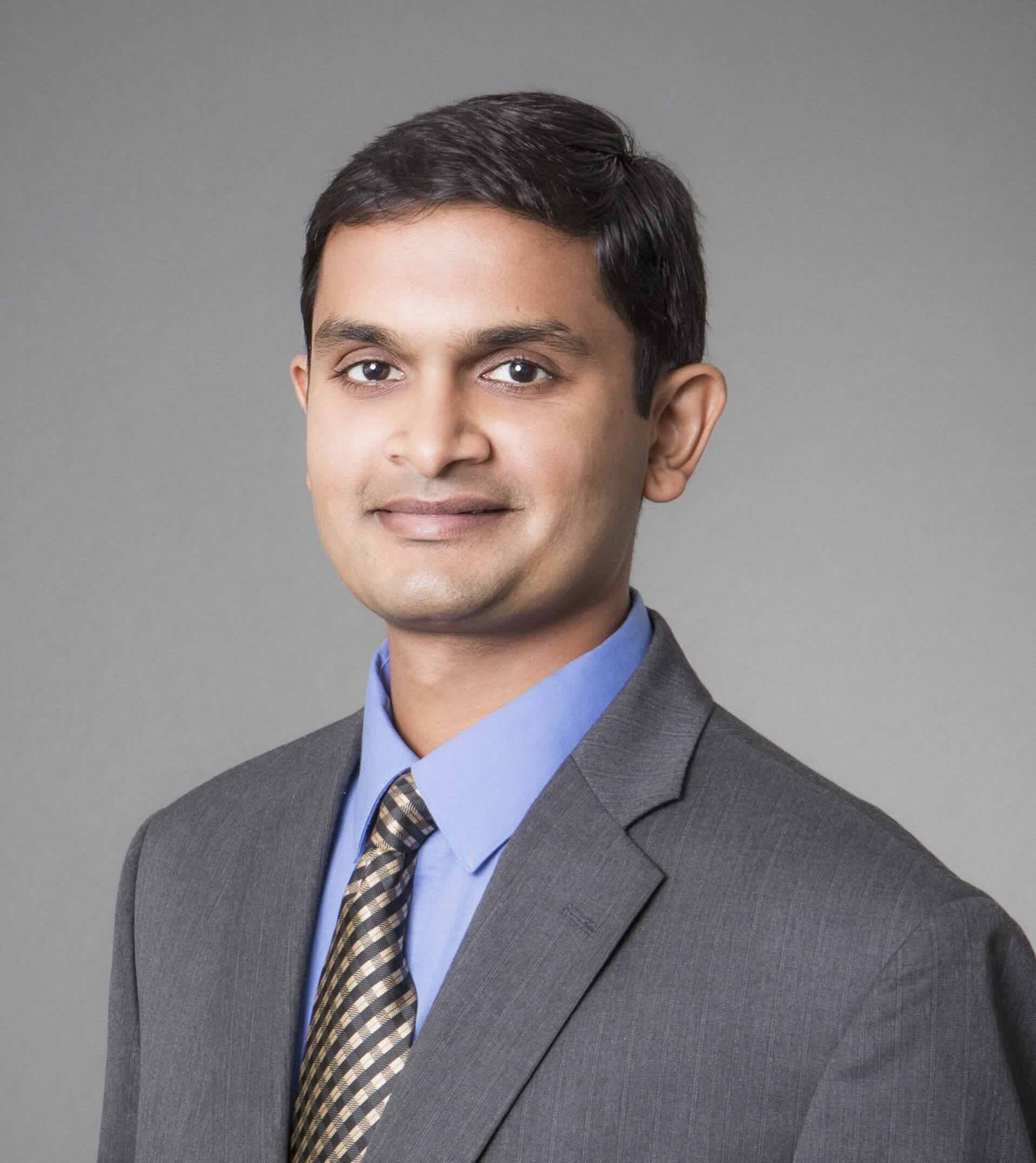}}]{Ketan Savla} is an associate professor and the John and Dorothy Shea Early Career Chair in Civil Engineering at the University of Southern California. His current research interest is in distributed optimal and robust control, dynamical networks, state-dependent queuing systems, and mechanism design, with applications in civil infrastructure and autonomous systems. His recognitions include NSF CAREER, IEEE CSS George S. Axelby Outstanding Paper Award, AACC Donald P. Eckman Award, and the IEEE ITSS Outstanding Application Award. He serve(d) as an associate editor of the IEEE Transactions on Control of Network Systems, IEEE Control Systems Letters (L-CSS), and IEEE Transactions on Intelligent Transportation Systems. He currently serves as a senior editor of the L-CSS. He is also a co-founder and the chief science officer of Xtelligent, Inc.  
	\end{IEEEbiography}
\end{document}